\documentclass[runningheads,a4paper,orivec]{llncs2e/llncs}
\usepackage{amssymb}
\usepackage{graphicx}

\usepackage{amsmath}
\usepackage{amsfonts}
\usepackage{amssymb}

\usepackage{url}
\urldef{\mailsa}\path|{alfred.hofmann, ursula.barth, ingrid.haas, frank.holzwarth,|
\urldef{\mailsb}\path|anna.kramer, leonie.kunz, christine.reiss, nicole.sator,|
\urldef{\mailsc}\path|erika.siebert-cole, peter.strasser, lncs}@springer.com|

\usepackage{blindtext}
\usepackage{url}
\usepackage{color}
\usepackage{nicefrac}
\usepackage{graphicx}
\usepackage{multirow}
\usepackage{rotate}
\usepackage{stackengine} % for \stackanchor
\usepackage{bm} % for vectors
\usepackage{enumerate}
\usepackage{thmtools}
\usepackage{thm-restate}
\usepackage{relsize} % Allow for the use of \mathlarger
\usepackage[noend]{algorithm2e} % For writing pseudo-code
\usepackage{multicol}
\usepackage{wrapfig}
\usepackage{ulem} 
\usepackage[bookmarks=false,linktocpage=true,colorlinks=false]{hyperref} % For use of \href{}.

\def\mywidth{.9}
\def\mywidthRep{.8}
\newif\ifcommentson\commentsonfalse
\def\mywidth{.9}
\def\mywidthRep{.8}
\ifcommentson
\newcommand{\commentCP}[1]{\begin{center} \parbox{\mywidth\textwidth}{\textbf{\textcolor{black}{Comment C.}} \textcolor{red}{#1 }}\end{center}}
\newcommand{\commentKC}[1]{\begin{center} \parbox{\mywidth\textwidth}{\textbf{\textcolor{black}{Comment K.}} \textcolor{red}{#1} }\end{center}}
\newcommand{\commentYK}[1]{\begin{center} \parbox{\mywidth\textwidth}{\textbf{\textcolor{black}{Comment Y.}} \textcolor{red}{#1} }\end{center}}
\newcommand{\commentMA}[1]{\begin{center} \parbox{\mywidth\textwidth}{\textbf{\textcolor{black}{Comment M.}} \textcolor{red}{#1} }\end{center}}
\newcommand{\replyCP}[1]{\begin{center} \parbox{\mywidthRep\textwidth}{\textbf{Reply C.} \textcolor{blue}{#1} }\end{center}}
\newcommand{\replyKC}[1]{\begin{center} \parbox{\mywidthRep\textwidth}{\textbf{Reply K.} \textcolor{blue}{#1} }\end{center}}
\newcommand{\replyYK}[1]{\begin{center} \parbox{\mywidthRep\textwidth}{\textbf{Reply Y.} \textcolor{blue}{#1} }\end{center}}
\newcommand{\replyMA}[1]{\begin{center} \parbox{\mywidthRep\textwidth}{\textbf{Reply M.} \textcolor{blue}{#1} }\end{center}}
\newcommand{\commentC}[1]{\marginpar{\footnotesize \color{red} {\bf C:} \textsf{\scriptsize #1}}}
\newcommand{\commentK}[1]{\marginpar{\footnotesize \color{red} {\bf K:} \textsf{\scriptsize #1}}}
\newcommand{\commentY}[1]{\marginpar{\footnotesize \color{red} {\bf Y:} \textsf{\scriptsize #1}}}
\newcommand{\commentM}[1]{\marginpar{\footnotesize \color{red} {\bf M:} \textsf{\scriptsize #1}}}
\newcommand{\replyC}[1]{\marginpar{\footnotesize \color{red} {\bf C:} \textsf{\scriptsize #1}}}
\newcommand{\replyK}[1]{\marginpar{\footnotesize \color{red} {\bf K:} \textsf{\scriptsize #1}}}
\newcommand{\replyY}[1]{\marginpar{\footnotesize \color{red} {\bf Y:} \textsf{\scriptsize #1}}}
\newcommand{\replyM}[1]{\marginpar{\footnotesize \color{red} {\bf M:} \textsf{\scriptsize #1}}}
\else
\newcommand{\commentCP}[1]{}
\newcommand{\commentKC}[1]{}
\newcommand{\commentYK}[1]{}
\newcommand{\commentMA}[1]{}
\newcommand{\replyCP}[1]{}
\newcommand{\replyKC}[1]{}
\newcommand{\replyYK}[1]{}
\newcommand{\replyMA}[1]{}
\newcommand{\commentC}[1]{}
\newcommand{\commentK}[1]{}
\newcommand{\commentY}[1]{}
\newcommand{\commentM}[1]{}
\newcommand{\replyC}[1]{}
\newcommand{\replyK}[1]{}
\newcommand{\replyY}[1]{}
\newcommand{\replyM}[1]{}

\fi
\newcommand{\calx}{\mathcal{X}}
\newcommand{\caly}{\mathcal{Y}}
\newcommand{\calc}{\mathcal{C}}

\newcommand{\cala}{\mathcal{A}}
\newcommand{\cald}{\mathcal{D}}

\newcommand{\cali}{\mathcal{I}}

\newcommand{\reals}{\mathbb{R}}
\newcommand{\distr}{\mathbb{D}}

\newcommand{\payd}{\mathit{u_{\sf d}}}
\newcommand{\paya}{\mathit{u_{\sf a}}}

\newcommand{\Payd}{\mathit{U_{\sf d}}}
\newcommand{\Paya}{\mathit{U_{\sf a}}}
\newcommand{\Pay}{\mathit{U}}
\newcommand{\eqdef}{\ensuremath{\stackrel{\mathrm{def}}{=}}}

\newcommand{\argmin}{\operatornamewithlimits{argmin}}
\newcommand{\supp}[1]{{\sf supp}(#1)}

\newcommand{\expectDouble}[2]{\operatornamewithlimits{\displaystyle\mathbb{E}}_{\substack{#1\\ #2}}}

\newcommand{\vf}{\mathbb{V}} % name of generic vulnerability function
\newcommand{\priorvf}[1]{\vf\!\left[#1\right]} % prior generiic vulnerabiltiy fuction
\newcommand{\postvf}[2]{\vf\!\left[#1,#2\right]} % posterior generiic vulnerabiltiy fuction

\newcommand{\samplefrom}{\leftarrow} 
\newcommand{\bigadd}{\operatorname{\sum}}

\newcommand{\dist}[1]{\mathbb{D}{#1}}

\newcommand{\qm}[1]{``#1''}
\newcommand{\true}{T}
\newcommand{\false}{F}

\newcommand{\smallsum}[1]{\textstyle{\sum_{#1}\:}}

\begin{document}

\mainmatter  % start of an individual contribution

\title{Information Leakage Games}

\titlerunning{Information Leakage Games}

\author{M\'{a}rio S. Alvim\inst{1}
\and Konstantinos Chatzikokolakis\inst{2}
\and Yusuke Kawamoto\inst{3} 
\and Catuscia Palamidessi\inst{4}}
\authorrunning{Alvim et al.}

\institute{Universidade Federal de Minas Gerais, Brazil
\and CNRS and \'{E}cole Polytechnique, France
\and AIST, Japan
\and INRIA and \'{E}cole Polytechnique, France}


\toctitle{Lecture Notes in Computer Science}
\tocauthor{Authors' Instructions}
\maketitle

\begin{abstract}{ %\scriptsize
We consider a game-theoretic setting to model the interplay between 
attacker and defender in the context of  information flow, 
and to reason about their optimal strategies.  
In contrast with standard game theory, in our games the utility of 
a mixed strategy is a convex function of the 
distribution on the defender's pure actions,  rather than the expected value of  their utilities.  
Nevertheless, 
%we show that 
the important properties of game theory, 
notably the existence of a Nash equilibrium, still hold for our (zero-sum) leakage games, 
and we provide algorithms to compute the corresponding optimal strategies. 
As typical in (simultaneous) game theory, the optimal strategy is usually mixed, i.e., 
probabilistic, for both the attacker and the defender. 
From the point of view of information flow, this  was to be expected in the case of the defender, 
since  it is well known  that randomization at the level of the system design may help to reduce information leaks. 
Regarding the attacker, however, this seems the first work (w.r.t. the literature in information flow)
proving formally  that in certain cases the optimal attack strategy is necessarily probabilistic. }
\end{abstract}

\section{Introduction}
\label{sec:introduction}
A fundamental problem in computer security is the leakage of sensitive 
information due to correlation of \emph{secret information} with 
\emph{observable information} 
%consisting in other information 
publicly available, 
or in some way accessible, to the attacker. 
Correlation in fact allows for the use of Bayesian inference 
to guessing the value of the secret. 
Typical  examples are \emph{side channels attacks}, 
in which (observable) physical aspects of the system, such as the execution 
time of a decryption algorithm, may be exploited  by the attacker to
restrict the range of the possible (secret) encryption keys. 
The branch of security that studies the amount of information leaked by a system is called 
\emph{Quantitative Information Flow} (QIF), and it has seen growing 
interest over the past decade.
%, including foundational works 
% \cite{Clark:05:JLC,Kopf:07:CCS,Chatzikokolakis:08:JCS,Smith:09:FOSSACS,McIver:10:ICALP,Alvim:12:CSF,Alvim:14:CSF,Boreale:15:LMCS,Alvim:16:CSF,Kawamoto:17:LMCS}, verification of information 
%flow properties \cite{Clark:07:JCS,Kopf:10:CSFa,Andres:10:TACAS,Chatzikokolakis:10:TACAS,ChothiaKawamoto2013,ChothiaKN14:esorics,Yasuoka:14:TCS,KawamotoBL16}, 
%detection of real system vulnerabilities \cite{Heusser:10:ACSAC,Kopf:12:CAV}, and, of course, methods to 
% reduce the information leakage. With respect to the latter, it has been recognized that  randomization 
% can be very useful  to obfuscate the link between secret and observables. 
%Examples include various anonymity protocols (for instance, the dining cryptographers \cite{Chaum:88:JC} and Crowds \cite{Reiter:98:TISS}), 
% methods for location privacy \cite{Shokri:12:CCS,Andres:13:CCS,Bordenabe:14:CCS}, and most notably,  the renown framework of differential privacy \cite{Dwork:06:TCC}.
See for instance \cite{Clark:07:JCS,Kopf:07:CCS,Smith:09:FOSSACS,Alvim:12:CSF,Boreale:15:LMCS}, just to mention a few.

In general, it has been recognized that  randomization 
can be very useful  to obfuscate the link between secrets and observables. 
Examples include various anonymity protocols (for instance, the dining cryptographers \cite{Chaum:88:JC} and Crowds \cite{Reiter:98:TISS}), and  the renown framework of differential privacy \cite{Dwork:06:TCC}.
The \emph{defender} (the system designer, or the user) is, therefore, typically probabilistic. 
As for the attacker,
% whose interest is to  maximize the leakage of information, 
most works in the literature consider only \emph{passive attacks}, 
%which are 
limited to observing the system's behavior. 
Notable exceptions are the works of Boreale and Pampaloni \cite{Boreale:15:LMCS},
and of Mardziel et al. \cite{Mardziel:14:SP}, which consider \emph{adaptive attackers} 
who interact with and influence the system.
%, possibly in an adaptive way. 
%\commentY{Boreale call them adaptive attackers.
%Cryptographers also call them adaptive attackers instead of interactive attackers.}
We note that, however, \cite{Boreale:15:LMCS} does not consider
probabilistic  strategies for the attacker.
%, i.e., attackers who can decide 
%the next move according to some probability distribution. 
%In fact, in their setting there is no need for an attacker to be 
%probabilistic: it is possible to show that for any probabilistic 
%strategy there is a deterministic one that induces at least as much 
%leakage. 
As for \cite{Mardziel:14:SP}, although their model allows them, 
none of their extensive  case-studies needs probabilistic attack strategies to maximize 
leakage.
This may seem surprising, since, as mentioned before, randomization is known to be 
useful (and, in general, crucial) for the defender to undermine the attack and protect the secret.
Thus there seems to be an asymmetry between attacker and defender w.r.t. 
probabilistic strategies in QIF. 
Our thesis is that there is indeed an asymmetry, but this does not mean that the 
attacker has nothing to gain from randomization: 
when the defender can change his own strategy according to the attacker's 
actions, it becomes advantageous for the attacker to try to be \emph{unpredictable}
and, consequently, adopt a probabilistic strategy. 
For the defender, while randomization is useful for the same reason, it is also 
useful because \emph{it reduces the information leakage},
and since information leakage constitutes the gain of the attacker, this reduction
influences his strategy.
This latter aspect introduces the asymmetry mentioned above. 

In the present work, we consider scenarios in which both attacker 
and defender can make choices that influence the system during the attack. 
We aim, in particular, at analyzing the attacker's strategies that can 
maximize information leakage, and the defender's most appropriate strategies to 
counterattack and keep the system as secure as possible. 
As argued before, randomization can help both attacker and defender 
make their moves unpredictable.  
The most suitable framework for analyzing this kind of interplay is, 
naturally, game theory, where the use of randomization can be modeled by 
the notion of \emph{mixed strategies}, and where the interplay 
between attacker and defender, and their struggle to achieve the best result for themselves, can be modeled  
in terms of \emph{optimal strategies} and \emph{Nash equilibrium}.
It is important to note, however, that one of the two advantages that randomization has for the defender, namely the reduction of information leakage, 
has no counterpart in standard game theory. 
Indeed, we demonstrate that this property makes the utility of a 
mixed strategy be a convex function of the distribution of the defender. 
In contrast, in standard game theory the utility of a mixed strategy is the expectation 
of the utility of the pure strategies of each player, and therefore it is an affine 
function on each of the players'  distributions.  
As a consequence, we need to consider a new kind of games, which we call 
\emph{information leakage games}, where the utility of a mixed strategy is a 
function affine on the attacker's strategy, and convex on the defender's. 
Nevertheless, 
%\commentC{I eliminated the sentence "we show that"}
the fundamental results of game theory, notably 
the minimax theorem and the existence of Nash equilibria, still hold for our 
zero-sum leakage
games. 
We also propose algorithms to compute the optimal strategy, namely, the 
strategies  for the attacker and the defender that lead to a Nash equilibrium, 
where no player has anything to gain by unilaterally changing his own strategy. 
%We therefore aim at developing a game-theoretic framework of
%\emph{information leakage games}.

For reasoning about information leakage, we employ the well-established
information-theoretic framework, which is by far the most used in QIF. 
A central notion in this model is that of 
%\emph{entropy}, but here
%we rather use its converse, 
\emph{vulnerability}, 
which 
intuitively measures how easily the secret can be discovered 
(and exploited) by the attacker. 
%In order to be as general as possible, we adopt 
%the very broad notion of vulnerability as any convex and continuous 
%function, as it has been used in \cite{Boreale:15:LMCS} and 
%\cite{Alvim:16:CSF}. 
For the sake of generality, we adopt 
the notion of vulnerability as any convex and continuous 
function \cite{Boreale:15:LMCS,Alvim:16:CSF}, 
which has been shown 
to subsume most previous measures of the QIF literature \cite{Alvim:16:CSF}, 
including \emph{Bayes vulnerability} (a.k.a. min-vulnerability~\cite{Smith:09:FOSSACS,Chatzikokolakis:08:JCS}), 
\emph{Shannon entropy}~\cite{Shannon:48:Bell}, 
\emph{guessing entropy}~\cite{Massey:94:IT}, and 
\emph{$g$-vulnerability}~\cite{Alvim:12:CSF}.

We note that vulnerability is an expectation measure over the secrets. 
In this paper we assume the utility   to be 
%(the converse of) 
such average measure, but,  in some cases, 
it could be advantageous for the defender to adopt different strategies depending on the value of the secret. 
We leave this  refinement for future work. 

The main contributions of this paper are the following: 
\begin{itemize}
\item We define a general framework of \emph{information leakage games}
to reason about the interplay between attacker and defender in 
QIF scenarios.
%, assuming they are both active players, can both influence 
%the system, and that the utility function of the game is the leakage 
%(vulnerability) of the system

\item We prove that, in our framework, the utility is  a convex function of the mixed 
strategy of the defender. 
%
%\item
To the best of our knowledge, this 
 is a novelty w.r.t. traditional game theory, 
where the utility of a mixed strategy is defined as expectation of the utilities 
of the pure strategies. 

\item We provide methods for finding the solution and the equilibria of 
 leakage games by solving a convex optimization problem. 
%\commentY{We could detail this a little by mentioning convex optimization and linear programming.}

\item We show examples in which Nash equilibria require a mixed strategy. 
This is, to the best of our knowledge, the first proof in 
%quantitative information flow 
QIF
that in some cases the optimal strategy of the attacker 
must be probabilistic.

%\item We show that nevertheless randomization still benefits
%primarily the defender. In fact, the computation of utility
%depends on whether or not the defender's actions are known 
%to the attacker, since vulnerability is a convex function 
%of the defender's strategy and in general randomized strategies
%reduce leakage.

\item As a case study, we consider the 
Crowds protocol in  a MANET (Mobile Ad-hoc NETwork).
%, in which nodes can only communicate with nearby nodes.
We study the case in which the attacker can add a 
corrupted node as an attack, the defender can add an honest 
node as a countermeasure, and we compute the defender  component of the   Nash equilibrium.
%investigate    
%each player's 
%best strategy when adding nodes to optimize leakage to their 
%benefit.
\end{itemize}

\paragraph*{Plan of the paper}
%The  paper is organized as follows.
In Section~\ref{sec:preliminaries} we review the basic notions 
of game theory and QIF.
In Section~\ref{sec:motivating-example} we introduce some motivating  examples. 
In section~\ref{sec:comparison-standart-gt}
we discuss the difference of our leakage games from those of standard game theory. 
In Section~\ref{sec:convex-games} we prove the convexity of the utility of the defender.
In Section~\ref{sec:nash-general}
we present algorithms for computing the Nash equilibria and 
optimal strategies for leakage games.
In Section~\ref{sec:case-study} we apply our framework to a version of the Crowds protocol.
In Section~\ref{sec:related-work} we discuss related work.	 
Section~\ref{sec:conclusion}  concludes.

\section{Preliminaries}
\label{sec:preliminaries}
In this section we review some basic notions from game theory
and QIF.

We use the following notation.
Given a set $\cali$, we denote 
% by $\#\cals$ its \emph{cardinality}, and 
by $\distr\cali$ the \emph{set of all probability distributions}
over $\cali$.
Given $\mu\in \distr\cali$, its \emph{support}
$\supp{\mu}$ is the set of its elements with positive probabilities, i.e.,
$\supp{\mu} = \{ i \in \cali : \mu(i)>0 \}$.
We write $i \samplefrom \mu$ to indicate that a value 
$i \in \cali$ is sampled from a distribution $\mu$ on $\cali$.

%Given a tuple $\vecC$ of $n$ elements, we denote by $| \vecC |$ the size of $\vecC$, i.e., $| \vecC | \eqdef n$.

\subsection{Two-player, simultaneous games}
\label{subsec:two-player-games}

We review basic definitions from \emph{two-player games},
a model for reasoning about the behavior of strategic players.
We refer to \cite{Osborne:94:BOOK} for more details. 
%The defender has some secret and tries not to leak it by his action while the attacker tries to reveal the defender's secret by performing some action.
%\commentMA{2017-02-03: Since they cover universal game-theory concepts,
%the preliminaries on game-theory should be as independent 
%as possible from our QIF-game set-up.
%We should avoid identifying the set of actions as stuff 
%such as \qm{hiding a secret} or \qm{revealing a secret}.
%To keep a consistent notation throughout the paper, we may name 
%players as \qm{defender} and \qm{attacker} in the preliminaries
%(instead of the more traditional \qm{player 1} and \qm{player 2}), 
%but we should point out that this happens w.l.o.g.}

In a game, each player has at its disposal a set of \emph{actions} 
that he can perform, and obtains some payoff (gain or loss)
depending on the outcome of the actions chosen by both players.
The payoff's value to each player is evaluated using a \emph{utility function}.
Each player is assumed to be \emph{rational}, i.e.,  
his choice is driven by the attempt to maximize his own utility.
%\commentMA{2017-02-03: The next sentence should not be here, but rather on the 
%section where we set-up the QIF games, justifying why game-theory 
%is a good choice for our model.}
%This is a fair assumption in the context of security because it is natural to figure out the best possible defending mechanism against the worst possible attacker.
We also assume that 
%We consider only games with \emph{complete information}, 
the set of possible actions and the utility functions
of both players are \emph{common knowledge}.
%, i.e., each player knows them, and he knows that the other knows, and he knows that the other knows that he knows, etc.  

In this paper we  only consider \emph{finite games}, namely the cases in which the 
set of actions available to each player is finite.   
%\commentMA{2017-02-10: Common knowledge has nothing to do with
%how many rounds or actions are there in a sequential game.
%It has to do with how each player \emph{reasons}.
%Think of a two-player game (call it Game A) in which both players
%are rational, but Player 1 believes that Player 2 \emph{is not rational}. 
%Now, imagine a two-player Game B in which both players
%are rational, Player 1 knows Player 2 is rational, but 
%Player 1 believes Player 2 does not know that Player 1 
%is rational.
%Player 1's behavior would be different in Game B than in Game A.
%As in Games A and B above, you can go stop at any level of the 
%ladder of \qm{Player 1 knows that 2 knows that 1 knows that 2 knowns...}, and at each level you stop you get a different behavior 
%of each player.
%\emph{Common knowledge} says that the you consider the ladder
%to be \emph{infinite}.
%}
%\replyYK{2017-02-11: You are right. The footnote looks good.}
%Next we introduce the important distinction  between \emph{simultaneous} and \emph{sequential} games. %
%In the following, we will call the two players \emph{defender} and \emph{attacker}.
%, respectively.
%
%\subsection{Simultaneous games}
%\label{subsec:simultaneous-games}
Furthermore, we only consider simultaneous games, meaning that each player chooses 
actions without knowing the actions chosen by the other.
%Note that   ``simultaneous'' does not mean that the
%players' actions are chosen at the same time, but that they are chosen independently.
Formally, such a game is defined as a tuple%
\footnote{Following the convention of \emph{security games}, we set the first player 
to be the defender.} 
$(\cald, \cala, \payd, \paya)$, where
%\begin{itemize}
%\item $\cald$ is a nonempty set of \emph{defender's actions}, 
%\item $\cala$ is a nonempty set of \emph{attacker's actions},
%\item $\payd: \cald\times\cala \rightarrow \reals$ is the 
%\emph{defender's utility function}, and
%\item $\paya: \cald\times\cala \rightarrow \reals$ is the 
%\emph{attacker's utility function}
%\end{itemize}
$\cald$ is a nonempty set of \emph{defender's actions}, 
$\cala$ is a nonempty set of \emph{attacker's actions},
$\payd: \cald\times\cala \rightarrow \reals$ is the 
\emph{defender's utility function}, and
$\paya: \cald\times\cala \rightarrow \reals$ is the 
\emph{attacker's utility function}.

Each player may choose an action deterministically or probabilistically.
A \emph{pure strategy} of the defender (resp. attacker) is a deterministic 
choice of an action, i.e., an element $d\in\cald$ (resp. $a\in\cala$). 
A pair $(d, a)$ is  a \emph{pure strategy profile}, and $\payd(d, a)$, $\paya(d, a)$ 
represent  the defender's and the attacker's utilities. 

A \emph{mixed strategy} of the defender (resp. attacker) is a probabilistic choice of an action, 
defined as a probability distribution 
$\delta\in\distr\cald$ (resp. $\alpha\in\distr\cala$).
A pair  $(\delta, \alpha)$ is called 
a \emph{mixed strategy profile}.
 %Note that $\delta$ and $\alpha$ are independent in simultaneous games
%(i.e., the distribution on each player's actions is independent of the
%other player's actions).
%
%To reason about mixed strategies we lift games from actions to probability distributions on actions.
%The \emph{mixed extension} of a game $\calg = (\cald,\, \cala,\, \payd, \paya)$ is the tuple $\widetilde{\calg} = (\distr\cald,\, \distr\cala,\, \Payd, \Paya)$, where $\distr\cald$ is the set of all distributions over defender's actions, $\distr\cala$ is the set of all distributions over attacker's actions, and $\Payd: \dist\cald\times\distr\cala \rightarrow \reals$ and $\Paya: \dist\cald\times\distr\cala \rightarrow \reals$ are the 
The defender's and the attacker's \emph{expected utility functions} for  mixed strategies are defined, 
respectively, as:
\begin{align*}
\Payd(\delta,\alpha) &
\eqdef {\displaystyle\expectDouble{d\leftarrow\delta}{a\leftarrow\alpha}\hspace{-0.5ex}  \payd(d, a)}
=\hspace{-0.5ex} \sum_{\substack{d\in\cald\\ a\in\cala}} \delta(d) \alpha(a) \payd(d, a) \\
\Paya(\delta,\alpha) & 
\eqdef {\displaystyle\expectDouble{d\leftarrow\delta}{a\leftarrow\alpha}\hspace{-0.5ex}  \paya(d, a)}
=\hspace{-0.5ex} \sum_{\substack{d\in\cald\\ a\in\cala}} \delta(d) \alpha(a) \paya(d, a)
\end{align*}

%We now present the notion of \qm{optimal strategies} and of \emph{Nash equilibrium}
%(for simultaneous games).
A defender's mixed strategy $\delta\in\distr\cald$ is a 
\emph{best response} to an attacker's mixed strategy 
$\alpha\in\distr\cala$ if 
$\Payd(\delta, \alpha) = \max_{\delta'\in\distr\cald}\Payd(\delta', \alpha)$.
Symmetrically, $\alpha\in\distr\cala$ is a \emph{best response} to $\delta\in\distr\cald$ 
if $\Paya(\delta, \alpha) = \max_{\alpha'\in\distr\cala}\Payd(\delta, \alpha')$.
A \emph{mixed-strategy Nash equilibrium} is a profile $(\delta^*, \alpha^*)$ such that $\delta^*$ is a best response to $\alpha^*$ and vice versa. 
%This means that in a Nash equilibrium, 
Namely, no unilateral deviation by any single player provides better 
utility to that player.
If $\delta^*$ and $\alpha^*$ are point distributions  concentrated on some $d^*\in\cald$ and $a^*\in\cala$, respectively, then 
$(\delta^*, \alpha^*)$ is a \emph{pure-strategy Nash equilibrium}, and will be denoted by $(d^*, a^*)$. 
%A defender's pure strategy $d\in\cald$ is a \emph{best response} 
%to an attacker's pure strategy $a\in\cala$ if 
%$\payd(d, a) = \max_{d'\in\cald}\payd(d', a)$.
%Symmetrically, $a\in\cala$ is a \emph{best response} to $d\in\cald$ 
%if $\paya(d, a) = \max_{a'\in\cala}\payd(d, a')$.
%A \emph{pure-strategy Nash equilibrium} of a simultaneous game $(\cald,\, \cala,\, \payd, \paya)$ is a  pure strategy profile $(d^*, a^*) \in \cald \times \cala$ s.t. $d^*$ is a best response to $a^*$ and that $a^*$ is a best response to $d^*$, i.e., for any $d\in\cald$, $\payd(d^*, a^*) \ge \payd(d, a^*)$, and for any $a\in\cala$, $\paya(d^*, a^*) \ge \paya(d^*, a)$.
%This means that at a Nash equilibrium,
%no unilateral deviations by any single player provide better payoff 
%to that player.
%
While not all games have a pure strategy Nash equilibrium, every finite game has a mixed strategy Nash equilibrium.
%A defender's mixed strategy $\delta\in\distr\cald$ is a 
%\emph{best response} to an attacker's mixed strategy 
%$\alpha\in\distr\cala$ if 
%$\Payd(\delta, \alpha) = \max_{\delta'\in\distr\cald}\Payd(\delta', \alpha)$.
%Symmetrically, $\alpha\in\distr\cala$ is a \emph{best response} to $\delta\in\distr\cald$ if $\Paya(\delta, \alpha) = \max_{\alpha'\in\distr\cala}\Payd(\delta, \alpha')$.
%%
%A \emph{mixed-strategy Nash equilibrium} of a simultaneous game $(\cald,\, \cala,\, \payd, \paya)$ is a mixed strategy profile $(\delta^*, \alpha^*) \in \distr\cald \times \distr\cala$ such that for any $\delta\in\distr\cald$, $\Payd(\delta^*, \alpha^*) \ge \Payd(\delta, \alpha^*)$, and for any $\alpha\in\distr\cala$, $\Paya(\delta^*, \alpha^*) \ge \Paya(\delta^*, \alpha)$.

\subsection{Zero-sum games and Minimax Theorem}
\label{subsec:zero-sum-games}
A game $(\cald,\, \cala,\, \payd, \paya)$ is \emph{zero-sum} if for any $d\in\cald$ and any $a\in\cala$,
$\payd(d, a) = -\paya(d, a)$, i.e., the defender's loss is equivalent to the attacker's gain.
For brevity, in zero-sum games we denote by $u$ the attacker's utility 
function $\paya$, and by $U$ the attacker's expected utility $\Paya$.%
\footnote{Conventionally in game theory  the utility $u$ is  set to 
be that of the first player, but we prefer to look at the utility 
from the point of view of the attacker to be in line with the definition of 
utility as \emph{vulnerability}, as we will introduce in Section~\ref{subsec:qif}.}
% A game is \emph{non-zero-sum} if it is not \emph{zero-sum}.
Consequently,  the goal of the defender is to minimize $U$, and the goal of the attacker is 
to maximize it. 

In simultaneous zero-sum games the Nash equilibrium corresponds to the solution of the \emph{minimax} problem (or equivalently,  the \emph{maximin} problem), namely, the profile $(\delta^*, \alpha^*)$ such that 
$U(\delta^*, \alpha^*)=\min_{\delta} \max_{\alpha} U(\delta, \alpha)$. 
The von Neumann's minimax theorem ensures that such solution (which always exists) is stable: 

\begin{theorem}[von Neumann's minimax theorem]
\label{theo:vonneumann}
Let $\calx \subset \reals^m$ and $\caly \subset \reals^n$ be compact convex sets,
and $\Pay: \calx\times\caly\rightarrow\reals$ be a continuous function such that
$\Pay(x, y)$ is convex in $x\in\calx$ and concave in $y\in\caly$.
Then it is the case that
%\[\min_{x\in\calx} \max_{y\in\caly} \Pay(x, y) = \max_{y\in\caly} \min_{x\in\calx} \Pay(x, y)\]
$\min_{x\in\calx} \max_{y\in\caly} \Pay(x, y) = \max_{y\in\caly} \min_{x\in\calx} \Pay(x, y)$.
\end{theorem}

A related property is that, under the conditions of Theorem~\ref{theo:vonneumann}, there exists a \emph{saddle point} $(x^*, y^*)$ s.t., for all $x\in\calx$  and $y\in\caly$, 
%\begin{align}
%\label{eq:saddle}
% \Pay(x^{*}, y) \leq  \Pay(x^*, y^*) \leq  \Pay(x, y^*) 
%\end{align}
$ \Pay(x^{*}, y){\leq}\Pay(x^*, y^*){\leq}\Pay(x, y^*)$.

%The solution of the minimax problem can be obtained by using convex optimization techniques. 
%In case $\Pay(x, y)$ is affine in $x$ and in $y$, we can also use linear optimization. 
%
%In the case in which  $\cald$ and $\cala$ contain two elements each, there is a closed form for the solution.  
%Let  $\cald =\{d_0,d_1\}$ and $\cala=\{a_0,a_1\}$ respectively. 
%Let $u_{ij}$ be the utility of the defender on $d_i, a_j$.           
%Then the Nash equilibrium $(\delta^*,\alpha^*)$ is given by:
%\begin{equation}
%\label{eq:closedform}
%\textstyle
%\delta^*(d_0)=\frac{ u_{11} - u_{10}}{  u_{00} - u_{01} - u_{10} +u_{11}} \;\;\;\;
%\alpha^*(a_0)=\frac{ u_{11} - u_{01}}{  u_{00} - u_{01} - u_{10} +u_{11}}
%\end{equation}
%if these values are in $[0,1]$.  
%Note that, since there are only two elements, the strategy $\delta^*$ is completely specified by its value in $d_0$, and analogously for $\alpha^*$.

\subsection{Quantitative information flow}
\label{subsec:qif}

Finally, we briefly review the standard framework of
quantitative information flow, 
which is used to measure the amount of information 
leakage in a system.
%which is concerned with
%measuring the amount of information leakage in a 
%(computational) system.
  
\paragraph{Secrets and vulnerability}
A \emph{secret} is some piece of sensitive information the 
defender wants to protect, such as a user's password, social 
security number, or current location. 
The attacker usually only has some partial knowledge 
about the value of a secret, represented as a probability
distribution on secrets called a \emph{prior}.
We denote by $\calx$ the set of possible secrets,
and we typically use $\pi$ to denote a prior belonging to 
the set $\dist{\calx}$ of probability distributions over 
$\calx$. 

The  \emph{vulnerability} of a secret is a measure of the utility 
of the attacker's knowledge about the secret. 
In this paper we consider a very general notion of vulnerability, 
following~\cite{Alvim:16:CSF}, and define a vulnerability $\vf$ 
to be any continuous and convex function of type 
$\dist{\calx} \rightarrow \reals$.
It has been shown in~\cite{Alvim:16:CSF} that these functions coincide
with the set of $g$-vulnerabilities, and are, in a precise sense, 
the most general information measures w.r.t. a set of basic 
axioms.~\footnote{
\label{footnote:convex}
More precisely, if posterior vulnerability 
is defined as the expectation of the vulnerability of posterior
distributions, 
%(as in Equation~\eqref{eq:postvf}), 
the measure respects the data-processing inequality 
and yields non-negative leakage iff
vulnerability is convex.}
%that we  shall review later on in these preliminaries.

\paragraph{Channels, posterior vulnerability, and leakage}
Systems can be modeled as information
theoretic channels.
A 
%\emph{channel matrix}, or simply a 
\emph{channel}
$C : \calx \times \caly \rightarrow \reals$ is a function
in which $\calx$ is a set of \emph{input values}, $\caly$ is a set 
of \emph{output values}, and $C(x,y)$ represents the conditional 
probability of the channel producing output $y \in \caly$ when 
input $x \in \calx$ is provided. 
Every channel $C$ satisfies $0 \leq C(x,y) \leq 1$ for all 
$x\in\calx$ and $y\in\caly$, and $\sum_{y\in\caly} C(x,y) = 1$ for all $x\in\calx$.

A distribution $\pi\in\dist{\calx}$ and a channel $C$ 
with inputs $\calx$ and outputs $\caly$ induce a joint distribution
$p(x,y) = \pi(x)C({x,y})$ on $\calx \times \caly$,
%producing joint random variables $X, Y$ 
with marginal 
probabilities $p(x) = \sum_{y} p(x,y)$ and 
$p(y) = \sum_{x} p(x,y)$, and conditional probabilities 
%$p(y{\mid}x) = \nicefrac{p(x,y)}{p(x)}$ (if $p(x)$ is non-zero) and
$p(x{\mid}y) = \nicefrac{p(x,y)}{p(y)}$ if $p(y) \neq 0$. 
% \commentYK{The following is not used anywhere in this paper.}
%Note that $p_{XY}$ is the unique joint distribution that recovers
%$\pi$ and $C$, in that $p(x) = \pi(x)$ and $p(y{\mid}x) = C({x,y})$ (if $p(x)$ is non-zero).\footnote{To avoid ambiguity, we may use subscripts on 
%distributions , e.g., $p_{XY}$, $p_{Y}$ or $p_{X \mid Y}$.}
For a given $y$ (s.t. $p(y) \neq 0$), the conditional 
probabilities $p(x{\mid}y)$ for each $x \in \calx$ form the 
\emph{posterior distribution $p_{X \mid y}$}.

A channel $C$ in which $\calx$ is a set of secret values 
and $\caly$ is a set of observable values produced
by a system can be used to model computations on secrets.
Assuming the attacker has prior knowledge $\pi$ about
the secret value, knows how a channel $C$ works, and
can observe the channel's outputs, the effect of the channel 
is to update the attacker's knowledge from a prior $\pi$ to a 
collection of posteriors $p_{X \mid y}$, each occurring 
with probability $p(y)$.%
%\footnote{This collection of posterior distributions is,
%in fact, a distribution on (posterior) distributions, 
%and is called a \emph{hyper-distribution} on secrets~\cite{McIver:10:ICALP}.}

Given a vulnerability $\vf$, a prior $\pi$, and a channel $C$, the \emph{posterior vulnerability} 
$\postvf{\pi}{C}$ is the vulnerability 
of the secret after the attacker has observed the output of 
$C$.
Formally:
%\begin{align}
%\label{eq:postvf}
%\postvf{\pi}{C}
%\eqdef&\, \sum_{y \in \caly} p(y) \priorvf{p_{X \mid y}}
%\end{align}
$\postvf{\pi}{C} \eqdef \sum_{y \in \caly} p(y) \priorvf{p_{X \mid y}}$.

%It is known from the literature~\cite{Alvim:16:CSF} 
%that  the posterior vulnerability is a convex function of $\pi$. 
%Namely, for any channel $C$, any family of distributions $\{\pi_i\}$, and any set of convex coefficients $\{c_i\}$, we have: 
%\begin{align}
%\label{eq:convexonpi}
%\postvf{\sum_i c_i \pi_i}{C}
%\leq &\, \sum_{i} c_i \postvf{\pi_i}{C}
%\end{align}
%

The \emph{information leakage} of 
a channel $C$ under a prior $\pi$ is a comparison between 
the vulnerability of the secret before the system
was run---called the \emph{prior} vulnerability---and the 
posterior vulnerability of the secret.
The leakage reflects by how much the observation of the 
system's outputs increases the utility of the attacker's 
knowledge about the secret. 
It can be defined either 
\emph{additively} ($\postvf{\pi}{C}-\priorvf{\pi}$), or
\emph{multiplicatively} ($\nicefrac{\postvf{\pi}{C}}{\priorvf{\pi}}$).

\section{A motivating example}
\label{sec:motivating-example}
We present some
 simple 
 %yet illustrative 
examples
to motivate our information leakage games. 

\subsection{The two-millionaires problem}
The ``two-millionaires problem'' was introduced by Yao in \cite{Yao:82:FOCS}. 
In the original formulation, there are two ``millionaires'', Alice and Don,
who want to discover who is the richest among them, but neither wants to 
reveal to the other the amount of money that he or she has. 

We consider a (conceptually) asymmetric variant of this problem, where Alice 
is the attacker and Don is the defender.
Don wants to learn whether or not he is richer than Alice, 
but does not want Alice to learn anything about the amount $x$ of money he has.
To this purpose, Don sends $x$ to a trusted server Jeeves, 
who in turn asks Alice, privately, what is her amount $a$ of money.
Jeeves then checks which among $x$ and $a$ is greater, and sends the 
result $y$ back to Don.%
\footnote{The reason to involve Jeeves is that Alice may not want to reveal $a$ to Don, either.}
However, Don is worried that Alice may intercept Jeeves' message containing the 
result of the comparison, and exploit it to learn more accurate information about $x$
by tuning her answer $a$ appropriately (since, given $y$, Alice can deduce whether
$a$ is an upper or lower bound on $x$).
We assume that Alice may get to know Jeeves' reply, 
but not the messages from Don to Jeeves. 

We will use the following information-flow terminology: 
the information that should remain secret (to the attacker) is called \emph{high}, 
and what is visible to (and possibly controllable by) the attacker is called \emph{low}. 
Hence, in the program run by Jeeves $a$ is a \emph{low input} 
and $x$ is a \emph{high input}. 
The result $y$ of the comparison (since it may be intercepted by the attacker) is a \emph{low output}. 
The problem is to avoid the \emph{flow  of information} from $x$ to $y$ (given $a$).

One way to mitigate this problem is to use randomization. 
Assume that Jeeves provides two different programs to ensure the service.
%\commentY{Added ``0 and 1'', because the value of $d$ was unclear from this sentence.}
Then, when Don sends his request to Jeeves, he can make a random choice $d$ 
among the two programs 0 and 1, sending $d$ to Jeeves along with the value $x$.
Now if Alice intercepts the result $y$, it will be less useful to her since
she does not know which of the two programs has been run.
%\commentY{Can we directly say ``Don knows program $d$ was run''?}
As Don of course knows which program was run, the result $y$ will still be 
just as useful to him.~\footnote{Note that $d$ should not be revealed to the attacker:
although $d$ is not sensitive information in itself, knowing it 
would help the attacker figure out the value of $x$.}

In order to determine the best probabilistic strategy that Don should apply to select the program, we  analyze the problem from a game-theoretic perspective. 
For simplicity, we assume that $x$ and $a$ both range in $\{0,1\}$. 
The two alternative programs that Jeeves can run are shown in Table~\ref{tab:jeeves}.
 
\begin{table}
\qquad
\qquad
\quad
\begin{small}
\begin{minipage}[t]{0.40\columnwidth}
\noindent 
\texttt{\underline{Program 0}}\\[2mm]
\noindent \texttt{\textbf{High Input:}} $x \in \{0,1\}$\\
\noindent \texttt{\textbf{Low Input:}} $a \in \{0,1\}$\\
\texttt{\textbf{Output:}} $y\in \{\true,\false\}$\\[1mm]
return $x\leq a$\\[1mm]
\end{minipage}
\quad
\begin{minipage}[t]{0.40\columnwidth}
\noindent 
\texttt{\underline{Program 1}}\\[2mm]
\noindent \texttt{\textbf{High Input:}} $x \in \{0,1\}$\\
\noindent \texttt{\textbf{Low Input:}} $a \in \{0,1\}$\\
\texttt{\textbf{Output:}} $y\in \{\true,\false\}$\\[1mm]
return $x\geq a$\\[1mm]
\end{minipage}
\end{small}
\caption{The two programs run by Jeeves.}\label{tab:jeeves}
\end{table}
\vspace{-8mm}
The combined choices of Alice and Don determine how the system behaves.
Let $\cald = \{0,1\}$ represent Don's possible choices, i.e., the program to run, 
and $\cala = \{0,1\}$ represent Alice's possible choices, i.e., the value of the 
low input $a$. 
We shall refer to the elements of $\cald$ and $\cala$ as \emph{actions}.
For each possible combination of actions 
$d$ and $a$, we can construct a channel 
$C_{da}$
with inputs $\calx = \{0,1\}$ (the set of possible high input values)
and outputs $\caly = \{\true, \false\}$ (the set of possible 
low output values),
modeling the behavior of the system \emph{from the point of view of the attacker}.
Intuitively, each channel entry $C_{da}(x,y)$ is the probability that 
the program run by Jeeves (which is determined by $d$) produces output 
$y\in\caly$ given that the high input is $x \in \calx$ and that the 
low input is $a$. 
The resulting four channel matrices are represented in Table~\ref{tab:two-millionaries-channels}.
Note that channels $C_{01}$ and $C_{10}$ do not leak any 
information about the input $x$ (output $y$ is constant), 
whereas channels $C_{00}$ and $C_{11}$ completely reveal $x$ (output $y$ 
is in a bijection with $x$).

\begin{table}[t]
\begin{center}
\begin{tabular}{c c c c c}
&
&
$a=0$
&
&
$a=1$
\\[-1ex]
\qquad
\\[-1ex]
$d = 0 \;\; \;(x \leq a?)$
&
\;\;
&
$\begin{array}{|c|c|c|}\hline
C_{00} & y=\true & y=\false \\ \hline
x=0    & 1 & 0 \\
x=1    & 0 & 1 \\ \hline
\end{array}
$
&
\;\;\;\;
&
$\begin{array}{|c|c|c|}
\hline
C_{01} & y=\true & y=\false \\ \hline
x=0    & 1 & 0  \\
x=1    & 1 & 0   \\ \hline
\end{array}
$
\\
\qquad
\\
$d = 1 \;\; \; (x \geq a?)$
&
\;\;
&
$
\begin{array}{|c|c|c|}
\hline
C_{10} & y=\true & y=\false \\ \hline
x=0    & 1 &  0\\
x=1    & 1 & 0 \\ \hline
\end{array}
$
&
\quad
&
$\begin{array}{|c|c|c|}
\hline
C_{11} & y=\true & y=\false \\ \hline
x=0    & 0 & 1 \\
x=1    & 1 & 0 \\ \hline
\end{array}
$
\end{tabular}
\end{center}
\caption{The two-millionaires system, from the point of view of the attacker.}
\vspace{-8mm}
\label{tab:two-millionaries-channels}
\end{table}

We want to investigate how the defender's and the attacker's strategies
influence the leakage of the system. 
For that we can consider the (simpler) notion of posterior vulnerability, 
since, for a given prior,
the value of leakage is in a one-to-one 
(monotonic) correspondence with the value of posterior vulnerability.
For this example, we consider posterior
Bayes vulnerability~\cite{Chatzikokolakis:08:JCS,Smith:09:FOSSACS}, defined as
%$\vf(\pi) = \max_{x} \pi(x)$ in the prior, and
$\postvf{\pi}{C} =\sum_y\max_xC(x,y)\pi(x)$.
Intuitively, Bayes vulnerability measures the probability of the 
adversary guessing the secret correctly in one try, and it
can be shown that $\postvf{\pi}{C}$
coincides with the converse 
%(i.e., the complement with respect to $1$) 
of the Bayes error.

For simplicity, we assume a uniform prior distribution $\pi_u$.
%In this case  the posterior Bayes vulnerability of a channel $C$ is the sum of the greatest elements 
%of each column of $C$, divided by the hight input domain size. Namely,
It has been shown that, in this case, the posterior Bayes vulnerability of
a channel $C$ can be computed as the sum of the greatest elements 
of each column of $C$, divided by the high input-domain size~\cite{Braun:09:MFPS}. 
Namely,
$\postvf{\pi_u}{C} = \nicefrac{\sum_y\max_xC(x,y)}{|\calx|}$.
It is easy to see that we have $\postvf{\pi_{u}}{C_{00}}= \postvf{\pi_{u}}{C_{11}} = 1$ and $\postvf{\pi_{u}}{C_{01}}= \postvf{\pi_{u}}{C_{10}} = \nicefrac{1}{2}$. 
Thus  we obtain the utility table shown in Table~\ref{table:vulnerabilitygame}, 
which is similar to that of the well-known \qm{matching-pennies} game.

\begin{wraptable}{r}{0.35\linewidth}
\centering
\vspace{-11mm}
\renewcommand{\arraystretch}{1.25}
\[
\begin{array}{|c|c|c|}
\hline
\vf & a=0 & a=1 \\ \hline
d=0    & 1 & \nicefrac{1}{2} \\ 
\hline
d=1    & \nicefrac{1}{2} & 1\\ \hline
\end{array}
\]
\renewcommand{\arraystretch}{1}
\vspace{-4mm}
\caption{{\rm Utility table for the two-millionaires game}.}
\label{table:vulnerabilitygame}
\vspace{-6mm}
\end{wraptable}
%\begin{table}[!htb]
%\centering
%\renewcommand{\arraystretch}{2}
%\[
%\begin{array}{|c|c|c|}
%\hline
%\vf & a=0 & a=1 \\ \hline
%d=     & 1 & 0 \\ 
%\hline
%d=     & 0 &  1\\ \hline
%\end{array}
%\]
%\renewcommand{\arraystretch}{1}
%%\vspace{2mm}
%\caption{{\rm The vulnerabilities of the channels $C_{da}$ in the running example} }
%\label{table:vulnerabilitygame}
%\end{table}

As in standard game theory, there may not exist an optimal pure strategy profile.
%$(d,a)$. 
The defender as well
as the attacker can then try to minimize/maximize the system's vulnerability by
adopting a mixed strategy $\delta$ and $\alpha$, respectively. A
crucial task is  \emph{evaluating the vulnerability} of the system under such
mixed strategies.
This evaluation is naturally performed from the point of view of the attacker,
who  knows his own choice $a$, but \emph{not the defender's choice
$d$}.
As a consequence, the attacker sees the system as the convex
combination $C_{\delta a} = \sum_d \delta(d)\, C_{ad}$, i.e., a probabilistic choice between
the channels representing the defender's actions. Hence, the overall
vulnerability of the system will be given by the vulnerability of $C_{\delta
a}$, averaged over all attacker's actions. 

We now define formally the ideas illustrated above.

\begin{definition}\label{eq:v-mixed}
An \emph{information-leakage game} is a tuple $(\cald, \cala, C)$ where $\cald, \cala$ are the sets of actions of the attacker and the defender, respectively, and  $C=\{C_{da}\}_{da}$ is a family of channel matrices indexed
on pairs of actions  $d\in\cald,a\in\cala$. For a given vulnerability $\vf$ and prior $\pi$, the utility of a pure strategy $(d,a)$
is given by $\postvf{\pi}{C_{d a}}$. The
utility $\vf(\delta,\alpha)$ of a mixed strategy $(\delta,\alpha)$ is  defined as:
\[
	\vf(\delta,\alpha)\eqdef \expectDouble{a\leftarrow\alpha}{}\hspace{-0.5ex} \postvf{\pi}{C_{\delta a}} =\smallsum{a} \alpha(a) \,\postvf{\pi}{C_{\delta a}}
	\quad\text{where}\quad
	C_{\delta a} \eqdef \sum_d \delta(d) \,C_{ad}
\]
\end{definition}

%Naturally, the attacker aims at maximizing the vulnerability of the system, 
%while the defender tries to minimize it. 
%None of the pure choices of $d$ and $a$ are the best for the corresponding player, 
%because the vulnerability of the system depends also on the (unknown) choice of the other player. 
%Yet there is a strategy leading to the best possible situation for both players (the \emph{Nash equilibrium}),
% but it is mixed, i.e., the players have to randomize their choices according to some precise distribution. 
 
%\begin{wraptable}{r}{0.35\linewidth}
%\centering
%\vspace{-7mm}
%$
%\begin{array}{|c|c|c|}
%\multicolumn{3}{c}{\text{Utility table for $a=0$}} \\[1mm]
%\hline
%C_{p0} & y=\true & y=\false \\ \hline
%x=0    & 1 & 0 \\
%x=1    & 1-p & p \\ \hline
%\end{array}
%$
%\\ \vspace{4mm}
%$
%\begin{array}{|c|c|c|}
%\multicolumn{3}{c}{\text{Utility table for $a=1$}} \\[1mm]
%\hline
%C_{p1} & y=\true & y=\false \\ \hline
%x=0    & p & 1-p  \\
%x=1    & 1 & 0   \\ \hline
%\end{array}
%$
%\caption{The two-millionaires mixed strategy of the defender, from the point of view of the attacker.}
%\label{tab:mixed two-millionaries-channels}
%\end{wraptable} 

In our example, $\delta$ is represented by a single number $p$: the
probability that the defender chooses $d = 0$ (i.e., Program $0$).
From
the point of view of the attacker, once he has chosen $a$, the system will look
like a channel $C_{pa} = p\, C_{0a} + (1-p)\,C_{1a} $. For instance, in the case $a=0$, if
$x$ is $0$ Jeeves will send $\true$ with probability $1$, but, if $x$ is $1$,
Jeeves will send $\false$ with probability $p$ and $\true$  with probability
$1-p$. Similarly for $a=1$.
Table~\ref{tab:mixed two-millionaries-channels} summarizes the various channels
modelling the attacker's point of view.
It is easy to see that $\postvf{\pi_{u}}{C_{p0}} = \nicefrac{(1+p)}{2}$ and $\postvf{\pi_{u}}{C_{p1}} = \nicefrac{(2-p)}{2}$. In this case $\postvf{\pi_{u}}{C_{pa}}$ coincides 
with the expected utility with respect to $p$, i.e., $\postvf{\pi_{u}}{C_{pa}} = p\,\postvf{\pi_{u}}{C_{0a}}+(1- p)\,\postvf{\pi_{u}}{C_{1a}}$. 
 
\begin{table}[t]
\centering
$
\begin{array}{|c|c|c|}
\multicolumn{3}{c}{\text{Utility table for $a=0$}} \\[1mm]
\hline
C_{p0} & y=\true & y=\false \\ \hline
x=0    & 1 & 0 \\
x=1    & 1-p & p \\ \hline
\end{array}
$
\qquad \qquad
$
\begin{array}{|c|c|c|}
\multicolumn{3}{c}{\text{Utility table for $a=1$}} \\[1mm]
\hline
C_{p1} & y=\true & y=\false \\ \hline
x=0    & p & 1-p  \\
x=1    & 1 & 0   \\ \hline
\end{array}
$
\vspace{2mm}
\caption{The two-millionaires mixed strategy of the defender, from the point of view of the attacker, where $p$ is the probability the defender picks action $d=0$.}
\label{tab:mixed two-millionaries-channels}
\vspace{-6mm}
\end{table}

Assume now that the attacker choses $a=0$ with probability $q$ and $a=1$ with probability $1-q$. The utility is obtained as expectation with respect to the strategy of the attacker, hence the total utility  is:
$\vf (p,q) = \nicefrac{q  \,(1+p)}{2} + \nicefrac{(1-p)\,(2-p)}{2}$, which is affine in both $p$ and $q$. By applying standard game-theoretic techniques, 
we  derive that the optimal strategy is $(p^*,q^*) = (\nicefrac{1}{2},\nicefrac{1}{2})$. 

In the above example,  things work just like in standard game theory. However, in the next section we will  show an example that fully exposes the difference of our games with respect to those of  standard game theory.  

\subsection{Binary sum}
The previous example is an instance of a general scenario in which a user, Don, delegates to a server, Jeeves, a certain computation that requires also some input from other users. 
Here we will consider another instance, in which the function to be computed is 
the binary sum $\oplus$.
% \colorR{(i.e., exclusive or)}. 
%\commentC{I eliminated the "exclusive or" explanation. The reviewer was confused because  y was declared with the wrong type (boolean) in the submitted version}
We assume Jeeves provides the programs in Table~\ref{tab:oplus}.
The resulting channel matrices are represented in Table~\ref{tab:oplus-channels}.
\begin{table}[h]
\qquad
\qquad
\quad
\begin{small}
\begin{minipage}[t]{0.40\columnwidth}
\noindent 
\texttt{\underline{Program 0}}\\[1mm]
\noindent \texttt{\textbf{High Input:}} $x \in \{0,1\}$\\
\noindent \texttt{\textbf{Low Input:}} $a \in \{0,1\}$\\
\texttt{\textbf{Output:}} $y\in \{0,1\}$\\ 
return $x\oplus a$\\[1mm]
\end{minipage}
\quad
\begin{minipage}[t]{0.40\columnwidth}
\noindent 
\texttt{\underline{Program 1}}\\[1mm]
\noindent \texttt{\textbf{High Input:}} $x \in \{0,1\}$\\
\noindent \texttt{\textbf{Low Input:}} $a \in \{0,1\}$\\
\texttt{\textbf{Output:}} $y\in \{0,1\}$\\
return $x\oplus a \oplus1$\\[1mm]
\end{minipage}
\end{small}
\caption{The two programs for $\oplus$ and its complement.}
\label{tab:oplus}
\vspace{-8mm}
\end{table}

\begin{table}[t]
\begin{center}
\begin{tabular}{r c c c c}
&
&
$a=0$
&
&
$a=1$
\\[-1ex]
\qquad
\\[-1ex]
$d = 0 \;\; \;(x \oplus a)$
&
\;\;
&
$\begin{array}{|c|c|c|}\hline
C_{00} & y=0 & y=1 \\ \hline
x=0    & 1 & 0 \\
x=1    & 0 & 1 \\ \hline
\end{array}
$
&
\;\;\;\;
&
$\begin{array}{|c|c|c|}
\hline
C_{01} & y=0 & y=1 \\ \hline
x=0    & 0 & 1  \\
x=1    & 1 & 0   \\ \hline
\end{array}
$
\\
\qquad
\\
$d = 1 \;\; \; (x \oplus a \oplus 1)$
&
\;\;
&
$
\begin{array}{|c|c|c|}
\hline
C_{10} & y=0 & y=1 \\ \hline
x=0    & 0 & 1\\
x=1    & 1 & 0 \\ \hline
\end{array}
$
&
\quad
&
$\begin{array}{|c|c|c|}
\hline
C_{11} & y=0 & y=1 \\ \hline
x=0    & 1 & 0 \\
x=1    & 0 & 1 \\ \hline
\end{array}
$
\end{tabular}
\end{center}
\caption{The binary-sum system, from the point of view of the attacker.}
\label{tab:oplus-channels}
\vspace{-4mm}
\end{table}

We consider again Bayes posterior vulnerability as utility. It is easy to see that we have $\postvf{\pi_{u}}{C_{00}}= \postvf{\pi_{u}}{C_{11}} = \postvf{\pi_{u}}{C_{01}}= \postvf{\pi_{u}}{C_{10}} = 1$. 
Thus  for the pure strategies we obtain the  utility table shown in Table~\ref{table:vulnerabilitygameoplus}. 
This means that all pure strategies have the same utility $1$ and therefore they are all equivalent.
In standard  game theory 
this would mean that also 
the mixed  strategies have
 the same utility $1$, since they are defined as expectation. 
In our case, however, the utility of a mixed strategy of the defender is convex on the
\begin{wraptable}{r}{0.35\linewidth}
\vspace{-4mm}
\centering
\renewcommand{\arraystretch}{1.25}
\[
\begin{array}{|c|c|c|}
\hline
\vf & a=0 & a=1 \\ \hline
d=0    & 1 & 1 \\ 
\hline
d=1    & 1 & 1\\ \hline
\end{array}
\]
\renewcommand{\arraystretch}{1}
\vspace{-6mm}
\caption{{\rm Utility table for the binary-sum game}.}
\label{table:vulnerabilitygameoplus}
\vspace{-6mm}
\end{wraptable}
%\begin{table}[!htb]
%\centering
%\renewcommand{\arraystretch}{2}
%\[
%\begin{array}{|c|c|c|}
%\hline
%\vf & a=0 & a=1 \\ \hline
%d=     & 1 & 0 \\ 
%\hline
%d=     & 0 &  1\\ \hline
%\end{array}
%\]
%\renewcommand{\arraystretch}{1}
%%\vspace{2mm}
%\caption{{\rm The vulnerabilities of the channels $C_{da}$ in the running example} }
%\label{table:vulnerabilitygame}
%\end{table}
distribution, so it may be convenient
 for the defender to adopt a mixed strategy. 
Let $p, 1-p$ be the probabilities of the defender choosing Program $0$ and Program $1$, respectively. From the point of view of the attacher, for each of his choices of $a$,   the system will appear as the probabilistic channel $C_{pa}$ represented in Table~\ref{tab:mixed binary sum}.
\begin{table}[h]
\begin{center}
\vspace{-4mm}
\begin{tabular}{ c c c}
$a=0$
&
&
$a=1$
\\[-1ex]
\qquad
\\[-1ex]
$\begin{array}{|c|c|c|}\hline
C_{p0} & y=\true & y=\false \\ \hline
x=0    & p & 1-p \\
x=1    & 1-p & p \\ \hline
\end{array}
$
&
\;\;\;\;
&
$\begin{array}{|c|c|c|}
\hline
C_{p1} & y=\true & y=\false \\ \hline
x=0    & 1-p & p  \\
x=1    & p & 1-p   \\ \hline
\end{array}
$
\end{tabular}
\end{center}
\caption{The binary-sum mixed strategy of the defender, from the point of view of the attacker, where $p$ is the probability the defender picks action $d=0$.}
\label{tab:mixed binary sum}
\vspace{-6mm}
\end{table}

It is easy to see that 
%\[
%\postvf{\pi_{u}}{C_{p0}} = \postvf{\pi_{u}}{C_{p1}} = 
%\left\{ 
%\begin{array} {ll}
%1-p & p \leq \nicefrac{1}{2}\\
%p & p \geq \nicefrac{1}{2}
%\end{array}
%\right.
%\]
$\postvf{\pi_{u}}{C_{p0}} = \postvf{\pi_{u}}{C_{p1}} = 1-p$ if $p \leq \nicefrac{1}{2}$,
and $\postvf{\pi_{u}}{C_{p0}} = \postvf{\pi_{u}}{C_{p1}} = p$ if $p \geq \nicefrac{1}{2}$.
On the other hand, with respect to a mixed strategy  of the attacker the utility is still defined as expectation. Since in this case the utility is the same for $a=0$ and $a=1$,   it  remains the same for any strategy of the attacker. Formally, $\vf(p,q) = q\, \postvf{\pi_{u}}{C_{p0}} + (1-q) \, \postvf{\pi_{u}}{C_{p1}} = \postvf{\pi_{u}}{C_{p0}}$, which does not depend on $q$ and it is minimum for  $p=\nicefrac{1}{2}$. We conclude that the point of equilibrium is $(p^*,q^*)=(\nicefrac{1}{2},q^*)$ for any value of $ q^*$.

\section{Leakage games vs. standard game theory models}
\label{sec:comparison-standart-gt}
In this section we explain the  differences between our information
leakage games and standard approaches to game theory.
We discuss:
(1) why the use of vulnerability as a utility function  makes 
our games non-standard w.r.t. von Neumann-Morgenstern's treatment of 
utility,
%(in which the utility of a mixed strategy is the expectation
%of the utility of pure strategies);
(2) why the use of concave utility functions to model risk-averse players
does not capture the behavior of the attacker in our games, and
(3) how our games differ from traditional convex-concave games.

\subsection{The von Neumann-Morgenstern's treatment of utility}

In their treatment of utility, 
von Neumann and Morgenstern~\cite{VonNeumann:47:Book}
demonstrated that the utility of a mixed strategy equals the expected utility 
of the corresponding pure strategies when a set of axioms is satisfied for 
player's preferences over probability distributions (a.k.a. \emph{lotteries}) on payoffs.
Since in our leakage games the utility of a mixed strategy is \emph{not} the
expected utility of the corresponding pure strategies, it is relevant to identify how
exactly our framework fails to meet von Neumann and Morgenstern (vNM) axioms.

Let us first introduce some notation.
Given two mixed strategies $\sigma$, $\sigma'$ for a player, we write
$\sigma \preceq \sigma'$ (or $\sigma' \succeq \sigma$) when the player 
prefers $\sigma'$ over $\sigma$, and $\sigma \sim \sigma'$
when the player is indifferent between $\sigma$ and $\sigma'$.
Then, the vNM axioms can be formulated as follows~\cite{Rubinstein:12:Book}.
For every mixed strategies $\sigma$, $\sigma'$ and $\sigma''$:
\begin{description}
\item[A1] \emph{Completeness}: it is either the case that
$\sigma \preceq \sigma'$, $\sigma \succeq \sigma'$, or $\sigma \sim \sigma'$.

\item[A2] \emph{Transitivity}: if $\sigma \preceq \sigma'$ and $\sigma' \preceq \sigma''$, 
then $\sigma \preceq \sigma''$.

\item[A3] \emph{Continuity}: if $\sigma \preceq\sigma'\preceq\sigma''$, then 
there exist $p \in [0,1]$ s.t. $p\,\sigma{+}(1-p)\,\sigma'' \sim\sigma'$.

\item[A4] \emph{Independence}: if $\sigma \preceq \sigma'$ then for any $\sigma''$
and   $p \in [0,1]$ we have $p\,\sigma + (1-p)\,\sigma'' \preceq p\,\sigma' + (1-p)\,\sigma''$.
\end{description}

%\begin{theorem}[von Neumann-Morgenstern utility theorem]
%For any rational player respecting axioms A1-A4 as above, there exists a utility
%function assigning to each outcome $C \in \mathcal{C}$ a payoff $u(C)$ such that 
%for any two mixed strategies $\sigma, \sigma' \in \dist{\mathcal{C}}$,
%\begin{align*}
%\sigma \preceq \sigma' \quad \text{iff} \quad \sum_{C \in \mathcal{C}} \sigma(C) u(C) \leq \sum_{C \in \mathcal{C}} \sigma'(C) u(C).
%\end{align*} 
%\end{theorem}

For any fixed prior $\pi $ 
%\in \dist{\calx}$ 
on secrets,  the utility function $u(C)=\vf[\pi,C]$ is a total function on $\calc$ 
ranging over the reals, and therefore it satisfies axioms A1, A2 and 
A3 above.
However, $u(C)$ does not satisfy A4, as the next example illustrates.
\begin{example}
%[Vulnerability does not satisfy independence]
Consider the following three channel matrices from input set $\calx = \{0,1\}$
to output set $\caly=\{0,1\}$, where $\epsilon$ is a small positive constant:
\begin{align*}
\begin{array}{|c|c|c|}
\hline
C_{1} & y = 0 & y = 1 \\ \hline
x = 0 & 1-\epsilon & \epsilon \\ 
x = 1 & \epsilon & 1-\epsilon \\ \hline
\end{array}
\qquad 
\begin{array}{|c|c|c|}
\hline
C_{2} & y = 0 & y = 1 \\ \hline
x = 0 & 1 & 0 \\ 
x = 1 & 0 & 1 \\ \hline
\end{array}
\qquad
\begin{array}{|c|c|c|}
\hline
C_{3} & y = 0 & y = 1 \\ \hline
x = 0 & 0 & 1 \\ 
x = 1 & 1 & 0 \\ \hline
\end{array}
\end{align*}
If we focus on Bayes vulnerability, it is clear that an attacker
would prefer $C_{2}$ over $C_{1}$, i.e., $C_{1} \preceq C_{2}$.
However, for the probability $p = \nicefrac{1}{2}$ we would have:
\begin{align*}
\begin{array}{|c|c|c|}
\hline
p\,C_{1} + (1-p)\,C_{3} & y = 0 & y = 1 \\ \hline
x = 0 & \nicefrac{(1-\epsilon)}{2} & \nicefrac{(1+\epsilon)}{2} \\
x = 1 & \nicefrac{(1+\epsilon)}{2} & \nicefrac{(1-\epsilon)}{2} \\ \hline
\end{array} 
\quad \text{and} \quad
\begin{array}{|c|c|c|}
\hline
p\,C_{2} + (1-p)\,C_{3} & y = 0 & y = 1 \\ \hline
x = 0 & \nicefrac{1}{2} & \nicefrac{1}{2} \\
x = 1 & \nicefrac{1}{2} & \nicefrac{1}{2} \\ \hline
\end{array}
\end{align*}
Since channel $p\,C_{1} + (1-p)\,C_{3}$ clearly reveals no less information about 
the secret than channel $p\,C_{2} + (1-p)\,C_{3}$, we have that
$p\,C_{1} + (1-p)\,C_{3} \succeq p\,C_{2} + (1-p)\,C_{3}$, and the axiom of
independence is not satisfied.
\end{example}
It is actually quite natural that vulnerability does not satisfy 
independence: a convex combination of two \qm{leaky} channels 
(i.e., high-utility outcomes) can produce a \qm{non-leaky} channel 
(i.e., a low-utility outcome). 
As a consequence, the traditional 
game-theoretic approach to the utility of mixed strategies does not 
apply to our information leakage games. However the existence of Nash equilibria is still granted, as we will see in Section~\ref{sec:convex-games}, Corollary~\ref{cor:Nash}.

%However the utility of a mixed strategy
%in our games is  convex  on the distribution of the defender (as it will be shown in Theorem~\ref{theo:convex-V-q}, Section~\ref{sec:convex-games}), and affine on the distribution of the adversary 
%(Definition~\ref{eq:v-mixed}), and hence the von Neumann's minimax
%theorem still applies, and the existence of a Nash equilibrium is guaranteed.

\subsection{Risk functions}

At a first glance, it may seem that our information leakage games 
could be expressed with some clever use of the concept of
\emph{risk-averse players} (in our case, the attacker), 
which is also based on convex utility functions (cf.~\cite{Osborne:94:BOOK}).
There is, however, a crucial difference: in the models of risk-averse
players, the utility function is convex 
\emph{on the payoff of an outcome of the game}, 
but the utility of a mixed strategy is still  
\emph{the expectation of the utilities of the pure strategies}, i.e., it is linear on the distributions.
On the other hand, the utility of mixed strategies in our information 
leakage games is \emph{convex on the distribution of the defender}.
This difference arises precisely because in our games utility
is defined as the vulnerability of the channel perceived by the
attacker, and, as we discussed, this creates an extra layer of
uncertainty for the attacker.

\subsection{Convex-concave games}

Another well-known model from standard game-theory is that of convex-concave
games, in which each of two players can choose among a continuous set of 
actions yielding convex utility for one player, and concave for the other.
In this kind of game the Nash equilibria are given by \emph{pure strategies} for 
each player. 
%This is because in these games for every mixed strategy 
%there is a pure strategy yielding the same utility.

A natural question would be why not represent our systems
as convex-concave games in which the pure actions of players are
the mixed strategies of our leakage games.
Namely, the real values $p$ and $q$ that uniquely determine the defender's 
and the attacker's mixed strategies, respectively, in the 
two-millionaires game of Section~\ref{sec:motivating-example}, could be 
taken to be the choices of pure strategies in a convex-concave  game
in which the set of actions for each player is the real interval $[0,1]$. 

This mapping from our games to convex-concave games, however, would not
be natural.
One reason is that utility is still defined as expectation in the standard 
convex-concave games, in contrast to our games. 
Consider two strategies $p_{1}$ and $p_{2}$ with utilities $u_1$ and $u_2$, respectively.
If we mix them using the coefficient $q\in [0,1]$, 
the resulting strategy 
 $q\, p_{1} + (1-q)\, p_{2}$
 will have utility  $u=q\, u_{1} + (1-q)\, u_{2}$ in the standard convex-concave game, 
 while in our case the utility would in general be strictly smaller than $u$.
The second reason is that a pure action corresponding to a mixed strategy may not always be realizable.
To illustrate this point, consider again the two-millionaires game, and the defender's mixed strategy consisting in choosing 
Program $0$ with probability $p$ and Program $1$ with probability $1-p$.
The requirement that the defender has a pure action corresponding to  $p$ implies the existence of a program (on Jeeves' side) that makes internally a probabilistic choice with bias $p$ and, depending on the outcome, executes Program $0$ or  Program $1$. However, it is not granted that Jeeves disposes of such a program. Furthermore, Don would not know what choice has actually been made, and thus the program would not achieve the same functionality, i.e., let Don know who is the richest. (Note that Jeeves should not communicate to Don the result of the choice, because of the risk that Alice intercepts it.) 
This latter consideration underlines a key practical aspect of leakage games,
namely, the defender's advantage over the attacker due to his knowledge 
of the result of his own random choice (in a mixed strategy). This advantage would be lost in a convex-concave representation of the game since the random choice would be ``frozen'' in its representation as a pure action.

\section{Convexity of vulnerability w.r.t. channel composition}
\label{sec:convex-games}
In this section we show that posterior vulnerability is a convex function of the strategy of the defender. 
In other words, given a set of channels, and a probability distribution over them, the 
vulnerability of the  composition of these channels according to the distribution is smaller than
or equal to the  composition of their vulnerabilities. 
As a consequence, we derive the existence of the Nash equilibria. 

In order to state this result formally, we introduce the following notation: given a channel matrix $C$ and a scalar $a$, $a\,C$ is the matrix obtained by multiplying every element of $C$ by $a$.
Given two \emph{compatible} channel matrices $C_1$ and $C_2$, namely matrices with the same indices of rows and columns\footnote{Note that two channel matrices with different column indices can always be made compatible by adding  appropriate columns with $0$-valued cells in each of them.}, $C_1+C_2$ 
%represents the matrix 
is obtained by adding  the cells of $C_1$ and $C_2$ with same indices. 
Note that if $\mu$ is a probability distribution on $\cali$, then 
 $ \bigadd_{i\in\cali} \mu(i)\,C_{i}$ is a channel matrix. 

\begin{theorem}[Convexity of vulnerability w.r.t. channel composition]
\label{theo:convex-V-q}
Let $\{C_{i}\}_{i \in \cali}$ be a family of compatible channels,
% all with domain $\calx$, 
and $\mu$ be a distribution on $\cali$.
Then, for every prior distribution $\pi$, and every
vulnerability $\vf$, the corresponding posterior vulnerability is convex w.r.t. to channel composition. Namely, for any probability distribution $\mu$ on $\cali$, we have
%\[\postvf{\pi}{\bigadd_{i} \mu(i)C_{i}} \leq \sum_{i} \mu(i) \,\postvf{\pi}{C_{i}}\]
$\postvf{\pi}{\bigadd_{i} \mu(i)\,C_{i}} \leq \sum_{i} \mu(i) \,\postvf{\pi}{C_{i}}$.
\end{theorem}
\begin{proof}
Define $p(y) = \sum_{x} \pi(x) \sum_{i} \mu(i)\,C_{i}(x,y)$.
Then:
\[
\begin{array}{rcll}
\postvf{\pi}{\bigadd_{i} \mu(i)\,C_{i}}
&=&\, \sum_{y  } p(y) \, \priorvf{\frac{\pi(\cdot)\, \sum_{i}\mu(i)\,C_{i}(\cdot,y) }{p(y)}}  & \text{(by def. of posterior $\vf$)}\\[2mm]
&=&\, \sum_{y } p(y) \, \priorvf{\sum_{i}\mu(i)\,\frac{\pi(\cdot)\, C_{i}(\cdot,y) }{p(y)}}  & \text{}\\[2mm] 
&\leq&\, \sum_{y } p(y) \, \sum_{i} \mu(i)\, \priorvf{\frac{\pi(\cdot)\, C_{i}(\cdot,y) }{p(y)}}  \;\;\;& \text{(*)}\\[2mm] 
&=&\, \sum_{i} \mu(i)\, \sum_{y} p(y)\, \priorvf{\frac{\pi(\cdot)\, C_{i}(\cdot,y) }{p(y)}}  & \text{}\\[2mm]
&=&\, \sum_{i} \mu(i)\, \postvf{\pi}{C_{i}}  & \text{(by def. of posterior $\vf$)}
\end{array}
\]
where (*) follows from the convexity of $\vf$ w.r.t. the prior  (cf.  Section \ref{subsec:qif}). 
\qed
\end{proof}

The existence of  Nash equilibria immediately follows from the above theorem:
\begin{corollary}\label{cor:Nash}
For any (zero-sum) information-leakage game there exist a Nash equilibrium, which in general is given by a mixed strategy.
%\commentC{I added ($0$-sum)}
\end{corollary} 
\begin{proof}
Given a mixed strategy $(\delta, \alpha)$, the utility $\vf(\delta, \alpha)$ given in Definition~\ref{eq:v-mixed}   is affine (hence concave) on $\alpha$.
Furthermore, by Theorem~\ref{theo:convex-V-q}, $\vf(\delta, \alpha)$ is convex on $\delta$. 
Hence we can apply the von Neumann's minimax
theorem (Section \ref{subsec:zero-sum-games}), which ensures the existence of a saddle point, i.e., a Nash equilibrium. 
\qed
\end{proof}

%has two important consequences: 
%From a practical point of view, it means that for the defender it is always convenient to adopt a mixed strategy, at least from the point of view of lowering the vulnerability of the system. 
%Of course there may be other reasons for not adopting a mixed strategy, or it may not be possible in certain scenarios, but from the sheer point of view of the vulnerability, a mixed strategy is always better. The exact distribution to choose for obtaining the best mixed strategy is, of course, 

%For each $i\in\cali$ let $q_i = \mu(i)$.
%Let $c(q_{i_1}, q_{i_2}, \ldots, q_{i_n}) \allowbreak = q_{i_1} C_{i_1} + q_{i_2} C_{i_2} + \ldots + q_{i_n} C_{i_n}$.
%Then $c$ is an affine function of $q_{i_1}, q_{i_2}, \ldots, q_{i_n}$.
%Since convexity is invariant under affine mappings, it follows from Proposition~\ref{prop:convex-Vg} that $\postvf{\pi}{c(q_{i_1}, q_{i_2}, \ldots, q_{i_n})}$ is a convex function of $q_{i_1}, q_{i_2}, \ldots, q_{i_n}$.
%Hence the proposition follows.
%%\commentMA{Kostas' results, which will be put in the 
%%Axioms paper's journal version.}
%%\replyYK{2017-02-03: I added Kostas' results in the previous proposition, and the proof for this proposition.}

%\section{Algorithms for NE and/or optimal strategy (simple version)}
%\label{sec:nash-simple}
%\input{nash-simple}

\section{Computing equilibria of information leakage games}
\label{sec:nash-general}

Our goal is to solve information leakage games, in which
the success of an attack $a$ and a defence $d$ is measured by a vulnerability
measure $\vf$. The attack/defence combination is a pure strategy profile $(d,a)$
in this game, and is associated with a channel $C_{da}$ modeling the behavior
of the system. The attacker clearly knows his own choice $a$, whereas the
defender's choice is assumed to be hidden. Hence
%as discussed in 
%Section~\ref{sec:motivating-example},
the utilty of a mixed strategy profile $(\delta,\alpha)$ will be given
by Definition~\ref{eq:v-mixed}, that is:
\[
	\vf(\delta,\alpha) = \smallsum{a}\alpha(a) \,\postvf{\pi}{\smallsum{d}\delta(d)\,C_{da}}
\]
%where $\pi$ is a prior distribution on the secrets.

Note that $\vf(\delta,\alpha)$ is convex on $\delta$ and affine on $\alpha$,
hence Theorem~\ref{theo:vonneumann} guarantees the existence of an equilibrium
(i.e. a saddle-point) $(\delta^*, \alpha^*)$ which is a solution of both the
minimax and the maximin problems. The goal in this section is to compute a) a
$\delta^*$ that is part of an equilibrium, which is important in order to
optimize the defence, and b) the utility $\vf(\delta^*,\alpha^*)$, which is
important to provide an upper bound on the effectiveness of an attack when
$\delta^*$ is applied.

This is a convex-concave optimization problem for which various methods have been
proposed in the literature. If $\vf$ is twice differentiable (and satisfies a
few extra conditions) then the Newton method can be applied \cite{Boyd:04:BOOK};
however, many such measures, most notably Bayes-vulnerability, our main
vulnerability measure of interest, are not differentiable. For non-differentiable
functions, \cite{Nedic:09:JOTA} proposes a subgradient method that
iterates on both $\delta,\alpha$ at each step. 
We have applied this method and it
does indeed converge to $\vf(\delta^*,\alpha^*)$, with one important caveat: the
solution $\delta$ that it produces is not necessarily an equilibrium (note that
$\vf(\delta,\alpha) = \vf(\delta^*,\alpha^*)$ does not guarantee that
$(\delta,\alpha)$ is a saddle point). Producing an optimal $\delta^*$ is of
vital importance in our case.

The method we propose is based on the idea of solving the minimax problem $\hat\delta = \argmin_\delta
\max_\alpha \vf(\delta,\alpha)$, since its solution is guaranteed to be part of
an equilibrium.\footnote{Note that this is true only for $\delta$, the
$\alpha$-solution of the minimax problem is not necessarily part of an
equilibrium; we need to solve the maximin problem for this.}
To solve this problem, we exploit the fact that %our function
$\vf(\delta,\alpha)$ is affine on $\alpha$ (not just concave). For a fixed $\delta$, maximizing
$\smallsum{a}\alpha(a) \,\postvf{\pi}{\smallsum{d}\delta(d)\,C_{da}}$ simply involves
picking the $a$ with the highest $\postvf{\pi}{\smallsum{d}\delta(d)\,C_{da}}$ and
assigning probability $1$ to it. Hence, our minimax problem is equivalent to
$\hat\delta = \argmin_\delta f(\delta)$ where
$f(\delta) = \max_a\postvf{\pi}{\smallsum{d}\delta(d)\,C_{da}}$; that
is, we have to minimize the max of finitely many convex functions, with $\delta$
being the only variables.

For this problem we can employ
the \emph{projected subgradient} method, given by:
\[
	\delta^{(k+1)} = P(\delta^{(k)} - \alpha_k g^{(k)})
\]
where $g^{(k)}$ is any subgradient of $f$
on $\delta^{(k)}$ \cite{Boyd:06:misc}.
Note that the subgradient of a finite max is simply a subgradient of
any branch that gives the max at that point.
$P(x)$ is the projection of $x$ on the domain of $f$; in our
case the domain is the probability simplex,
for which there exist efficient algorithms for computing the projection
\cite{Wang:2013:arXiv}. Finally $\alpha_k$ is a step-size, for which various
choices guarantee convergence \cite{Boyd:06:misc}. In our experiments we
found $\alpha_k = 0.1/\sqrt{k}$ to perform well.

As the starting point $\delta^{(1)}$ we take the uniform distribution; moreover
the solution can be approximated to within an arbitrary $\epsilon> 0$
by using the stopping criterion of \cite[Section~3.4]{Boyd:06:misc}.
Note that the obtained $\hat\delta$ approximates the equilibrium strategy
$\delta^*$, while $f(\hat\delta)$ approximates $\vf(\delta^*,\alpha^*)$. Hence
we achieve both desired goals, as formally stated in the following result.

\begin{proposition}
	If $\vf$ is Lipschitz then
	the subgradient method discussed in this section converges to a $\delta^*$
	that is part of an equilibrium of the game.
	Moreover, let $\hat\delta$ be the solution computed within a given
	$\epsilon >0$, and let $(\delta^*, \alpha^*)$ be an equilibrium.
	Then it holds that:
	\[
		\vf(\hat\delta, \alpha) - \epsilon
		\le
		\vf(\hat\delta,\alpha^*)
		\le
		\vf(\delta,\alpha^*) + \epsilon
		\qquad \forall \delta,\alpha
	\]
	which also implies that $f(\hat\delta) - \vf(\delta^*,\alpha^*) \le \epsilon$.
\end{proposition}

\begin{proof}(\emph{Sketch}) $\argmin_\delta f(\delta)$ is equivalent to the minimax problem whose
$\delta$-solution is guaranteed to be part of an equilibrium. Convergence is
ensured by the subgradient method under the Lipschitz condition, and given that
$||\delta^{(1)}-\delta^*||$ is bounded by the distance
between the uniform and a point distribution.
\qed
\end{proof}

Finally, of particular interest is the Bayes-vulnerability measure
\cite{Chatzikokolakis:08:JCS,Smith:09:FOSSACS}, given by 
$\postvf{\pi}{C} = \smallsum{y}\max_x \pi(x) \,C(x,y)$, since it is widely used to
provide an upper bound to all other measures of information
leakage~\cite{Alvim:12:CSF}.
For this measure, $\vf$ is Lipschitz and the subgradient vector $g^{(k)}$ is given by
$
g^{(k)}_d =   \sum_y  \pi(x^*_y)\, C_{da^*}(x^*_y,y)
$
where $a^*,x^*_y$ are the ones giving the max in the branches of
$f(\delta^{(k)})$.
Note also that, since $f$ is piecewise linear, the convex optimization problem
can be transformed into a linear one using a standard technique, and then solved
by linear programming. 
However, due to the large number of max branches
of $\vf$, this conversion can be a problem with a huge number
of constraints.
In our experiments we found that the subgradient method
described above is significantly more efficient than linear programming.
%\commentKC{Address rev.2: ``what is the problem by
%linearizing the piecewise linear function?''}

Note also that, although the subgradient method is general, it might
be impractical in applications where the number of attacker or defender actions
is very large. Application-specific methods could offer better scalability in
such cases, we leave the development of such methods as future work.
%\commentKC{Address rev.2: `` what if the space for defender and attacker actions
%are exponentially large so that directly solving the game via minimax becomes
%computationally prohibited.''}

%\max_x  \sym_d delta(d) \sum_y \max_x\pi(x) C^ad(x,y)
% g^k_d =   \sum_y  \pi(x^*_a) C^a^*d(x^*_a,y)

\section{Case study}
\label{sec:case-study}
In this section, we apply our game-theoretic analysis to the case of anonymous
communication on a mobile ad-hoc network (MANET). 
In such a network, nodes can move in space and communicate with other nearby nodes. 
We assume that nodes can also access some global 
(wide area) network, but such
connections are neither anonymous nor trusted. Consider, for instance, smartphone
users who can access the cellular network, but do not trust the network
provider. The goal is to send a message on the global network without revealing
the sender's identity to the provider. For that, users can form
a MANET using some short-range communication method (e.g., bluetooth), and take
advantage of the local network to achieve anonymity on the global one.

Crowds~\cite{Reiter:98:TISS} is a protocol for anonymous communication that can
be employed on a MANET for this purpose.
Note that, although more advanced systems for anonymous communication
exist (e.g. Onion Routing), the simplicity of Crowds makes it particularly
appeling for MANETs.
%\commentKC{Address rev.2: ``
%- How realistic is Crowds protocol? Are there any other protocols? 
%''}
The protocol works as follows:
the \emph{initiator} (i.e., the node who wants to send the message) selects
some other node connected to him (with uniform probability) and forwards the
request to him. A \emph{forwarder}, 
upon receiving the message, performs a
probabilistic choice: with probability $p_f$ he keeps forwarding the message
(again, by selecting uniformly a user among the ones connected to him), while
with probability $1-p_f$ he delivers 
the message on 
the global network. 
Replies,
if any, can be routed back to the initiator following the same path in reverse
order.

Anonymity comes from the fact that the \emph{detected} node (the last in the
path) is most likely not the initiator. Even if the attacker knows the network
topology, he can infer that the initiator is most likely a node close to the
detected one, but if there are enough nodes we can achieve some reasonable
anonymity guarantees. However, the attacker can gain an important advantage by
deploying a node himself and participating to the MANET. When a node forwards a
message to this \emph{corrupted} node, this action is observed by the attacker
and increases the probability of that node being the initiator. Nevertheless,
the node can still claim that he was only forwarding the request for someone
else, hence we still provide some level of anonymity. By modeling the system as
\begin{wrapfigure}[15]{r}{0.38\linewidth}
	\vspace{-6mm}
	\centering
	\includegraphics[width=0.38\columnwidth]{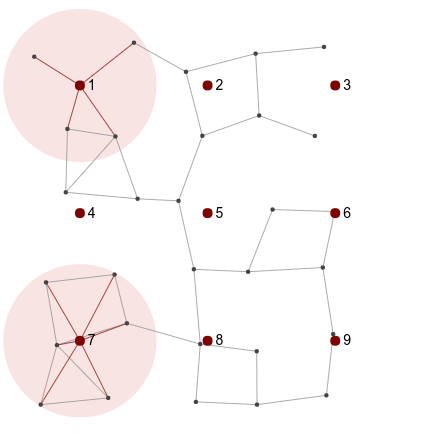}
	\vspace{-7mm}
	\caption{A MANET with 30 users in a $1$km$\times 1$km area.}
	\label{fig:manet}
\end{wrapfigure}
a channel, and computing its posterior Bayes vulnerability
\cite{Smith:09:FOSSACS}, we get the probability that the attacker guesses
correctly the identity of the initiator, after performing his observation.

In this section we study a scenario of 30 nodes deployed in an area of
$1$km$\times 1$km, in the locations illustrated in Fig.~\ref{fig:manet}. Each
node can communicate with others up to a distance of 250 meters, forming the
network topology shown in the graph. To compromise the anonymity of the system,
the attacker plans to deploy a corrupted node in the network; the question is
which is the \emph{optimal location} for such a node. The answer is far from trivial:
on the one side being connected to many nodes is beneficial,  but at the same
time these nodes need to be ``vulnerable'', being close to a highly connected
clique might not be optimal.
%\commentM{Does Table~\ref{table:utilities} show vulnerabilities or leakage?
%Values are $> 1$!}
% Kostas: it's vulnerability, just expressed in % to avoid very small numbers
At the same time, the administrator of the network is suspecting that the
attacker is about to deploy a corrupted node. Since this action cannot be
avoided (the network is ad-hoc), a countermeasure is to deploy a
\emph{deliverer} node at a location that is most vulnerable.
Such a node directly delivers all messages forwarded to it on the global
network; since it
never generates messages its own anonymity is not an issue, it only
improves the anonymity of the other nodes. Moreover, since it never communicates
in the local network its operation is invisible to the attacker.
But again, 
the optimal location for the new deliverer node is not obvious, and most importantly,
the choice depends on the choice of the attacker.

To answer these questions, we model the system as a game where the actions of 
attacker and defender are the locations of newly deployed corrupted
and honest nodes, respectively. We assume that the 
possible locations for new nodes are the nine ones shown in
Fig.~\ref{fig:manet}. For each pure strategy profile $(d,a)$, we construct the
corresponding network and use the PRISM model checker to construct the
corresponding channel $C_{da}$, using a model similar to the one
of~\cite{Shmatikov:02:CSFW}.
Note that the computation considers the specific network topology
of Fig.~\ref{fig:manet}, which reflects the positions of each node at the time
when the attack takes place; the corresponding channels need to be recomputed if
the network changes in the future.
As  leakage measure we use the
posterior Bayes vulnerability (with uniform prior $\pi$), which is the
attacker's probability of correctly guessing the initiator given his
observation in the protocol. According to Definition~\ref{eq:v-mixed}, for a mixed
%\commentKC{Address rev.2: ``
%how did you compute the C(x,y) since the users are moving in the network and there is a stochastic process during the procedure by randomly influencing neighbor nodes.
%''}
%where $\pi$ is a prior distribution modelling how frequently each node sends
%messages on the network, assumed to be uniform for the purposes of this case
%study.
%we can compute its
%posterior Bayes vulnerability (assuming a uniform prior), which gives the
%attacker's probability of correctly guessing the initiator, given his
%The Bayes vulnerability is taken as the attacker's utility in the game;
\begin{wraptable}[13]{r}{0.56\linewidth}
%\begin{table}[!tb]
	\vspace{-8mm}
	\centering
	\begin{scriptsize}
	\[\def\arraystretch{1.2}
	 \begin{array}{c|c|ccccccccc|}
	  \multicolumn{2}{c}{}  & \multicolumn{9}{c}{\text{\small Attacker's action}} \\ \cline{2-11}
	  &   & \textbf{1} & \textbf{2} & \textbf{3} & \textbf{4} & \textbf{5} & \textbf{6} & \textbf{7} & \textbf{8} & \textbf{9}  \\ \cline{2-11}
	  \multirow{9}{*}{\rotatebox[origin=c]{90}{\text{\small Defender's action}}}
	  & \textbf{1} & 7.38 & 6.88 & 6.45 & 6.23 & 7.92 & 6.45 & 9.32 & 7.11 & 6.45 \\
	  & \textbf{2} & 9.47 & 6.12 & 6.39 & 6.29 & 7.93 & 6.45 & 9.32 & 7.11 & 6.45 \\
	  & \textbf{3} & 9.50 & 6.84 & 5.46 & 6.29 & 7.94 & 6.45 & 9.32 & 7.11 & 6.45 \\
	  & \textbf{4} & 9.44 & 6.92 & 6.45 & 5.60 & 7.73 & 6.45 & 9.03 & 7.11 & 6.45 \\
	  & \textbf{5} & 9.48 & 6.91 & 6.45 & 6.09 & 6.90 & 6.13 & 9.32 & 6.92 & 6.44 \\
	  & \textbf{6} & 9.50 & 6.92 & 6.45 & 6.29 & 7.61 & 5.67 & 9.32 & 7.11 & 6.24 \\
	  & \textbf{7} & 9.50 & 6.92 & 6.45 & 5.97 & 7.94 & 6.45 & 7.84 & 7.10 & 6.45 \\
	  & \textbf{8} & 9.50 & 6.92 & 6.45 & 6.29 & 7.75 & 6.45 & 9.32 & 6.24 & 6.45 \\
	  & \textbf{9} & 9.50 & 6.92 & 6.45 & 6.29 & 7.92 & 6.24 & 9.32 & 7.11 & 5.68 \\\cline{2-11}
	\end{array}
	\vspace{-4ex}
	\]
	\end{scriptsize}
	\vspace{-2mm}
	\caption{Utility for each pure strategy profile.}
	\label{table:utilities}
	\vspace{-10mm}
%\end{table}
\end{wraptable}
%\begin{table}[!tb]
%\centering
%	\[\def\arraystretch{1.5}
%	 \begin{array}{|c|ccc|}
%	\hline
%	 & \textbf{Simultaneous}
%	 & \textbf{Def. 1\textsuperscript{st}}
%	 & \textbf{Adv. 1\textsuperscript{st}} \\ \hline
%	 \textbf{Def. best strategy}
%	 	& 1\,(46.8\%) /\,7 (53.2\%)
%		& 7
%		& 7
%		\\
%	 \textbf{Adv. best strategy}
%	 	& 1\,(40.9\%) / 7\,(59.1\%)
%		& 1
%		& 7
%	 	\\
%	 \textbf{Utility}
%	 	& 8.69\%
%		& 9.50\%
%		& 8.15\%
%	 	\\
%		\hline
%	\end{array}
%	\]
%	\caption{Best strategies and utilities for each game.}
%	\label{table:results}
%\end{table}
strategy profile $(\delta,\alpha)$ the utility is 
$\vf(\delta,\alpha)= \expectDouble{a\leftarrow\alpha}{}\hspace{-0.5ex}\postvf{\pi}{C_{\delta a}}$.

The utilities  (posterior Bayes vulnerability \%) for each pure profile are displayed in Table~\ref{table:utilities}.
Note that the attacker and defender actions substantialy affect the
effectiveness of the attack, with the probability of a correct guess ranging
between $5.46\%$ and $9.5\%$.
%\commentKC{Address rev.2: ``
%- The Utility in Table 9 seems to not vary much in each column.  
%''}
Based on the results of 
Section~\ref{sec:nash-general},
we can then compute the best strategy for the defender, which turns out to be (probabilities expressed as \%):
\[
	\delta^* = (34.59,  3.48,  3.00,  10.52,  3.32,  2.99,  35.93,  3.19,  2.99)
\]
This strategy is part of an equilibrium %(a saddle point of $\vf(\delta, \alpha)$) 
and guarantees
that for any choice of the attacker the vulnerability is at most $8.76\%$, and
is substantially better that the best pure strategy (location~1) which leads to
a worst vulnerability of $9.32\%$. As expected, $\delta^*$ selects the most
vulnerable locations (1 and 7) with the highest probability. 
Still, the
other locations are  selected with non-negligible probability, which is
important for maximizing the attacker's uncertainty about the defense.

%For simplicity, we only consider the case when the defender's choice is visible,
%which is reasonable in our scenario since the location of each node in a MANET
%can generally be learned. The results are summarized in Table~\ref{table:results}.
%We notice that the two more interesting locations for inserting new nodes are
%1 and 7, highlighted in Fig.~\ref{fig:manet} (the range and
%available connections for these two locations are also shown in the figure).
%Naturally, the defender would like to place his node at the same location as the
%attacker's node, while the latter wishes the opposite.

%When the defender plays first (Game II), his best strategy is location 7, which
%leads to the attacker selecting location 1, giving a posterior vulnerability of
%9.5\%. Conversely, if the attacker plays first (Game III), his best choice is 7,
%and the defender's best option is to match it giving a vulnerability of 8.15\%.
%Finally, if both players select simultaneously (Game I), then neither of them
%has inventive to commit to a single choice, leading to a mixed equilibrium where
%locations 1 and 7 are selected with different probabilities by each player,
%giving a posterior vulnerability of 8.69\%. As expected by the discussion of
%Section~\ref{sec:comparing-games}, the utility in the simultaneous case lies
%between the utilities of the other two cases

\section{Related work}
\label{sec:related-work}
%Many works have used game theory to analyze security 
%and privacy in computational systems,
%including areas such as
%jamming at the the MAC layer~\cite{Basar:83:TIT}, 
%ephemeral networks~\cite{Raya:08:CCS}, 
%vehicular networks~\cite{Alpcan:11:TMC}, 
%network security~\cite{Grossklags:08:WWW},
%cryptography~\cite{Katz:08:TCC},
%anonymity~\cite{Acquisti:03:FC}, 
%location privacy~\cite{Freudiger:09:CCS}, and
%intrusion detection systems~\cite{Zhu:09:GAMENETS}, 
%to cite a few.
There is an extensive literature on game theory models for security 
and privacy in computer systems,
% Note that {Manshaei:13:ACMCS} is a survey paper
% I am not sure if we should cite it
including 
network security,
%areas such as
%jamming at the 
%the MAC layer, 
%ephemeral networks, 
vehicular networks, 
cryptography,
%\commentY{Cryptography is not limited to the purpose of network security, and its community is different from those of the other applications written here.}
anonymity, 
location privacy, and
intrusion detection. See \cite{Manshaei:13:ACMCS} for a survey. 

In many studies, security games have been used to model and analyze utilities 
between interacting agents, especially an attacker and a defender.
In particular, Korzhyk et al. \cite{Korzhyk:11:JAIR} present 
a theoretical analysis of security games and
investigate the relation between Stackelberg and 
simultaneous games under various forms of uncertainty.
In application to network security,
Venkitasubramaniam~\cite{Venkitasubramaniam:12:ACMTN} investigates 
% the optimization of anonymity in 
anonymous wireless networking, which they formalize as a  
zero-sum game between the network designer and the attacker.
The task of the attacker is to choose a subset of nodes to monitor 
so that anonymity of routes is minimum whereas the task of the 
designer is to maximize anonymity by choosing nodes to 
evade flow detection by generating independent transmission schedules.
%In this two player game, it is shown that a unique saddle point equilibrium exists for a general category of finite networks.
%Anonymity is measured as conditional entropy of the routes given the attacker’s observation.

%Security games have been used to find optimal 
%trade-offs between two conflicting desirable properties.
%As for usability, 
Khouzani et al.~\cite{Khouzani:15:CSF} present a  
framework for analyzing a trade-off between usability and 
security.
They analyze guessing attacks and derive the optimal policies 
for secret picking as Nash/Stackelberg equilibria.
Khouzani and 
Malacaria \cite{Khouzani:16:CSF} investigate properties of 
leakage when perfect secrecy is not achievable due to the 
limit on the allowable size of the conflating sets, and
show the existence of universally optimal strategies for
a wide class of entropy measures, and for $g$-entropies 
(the dual of $g$-vulnerabilities).
In particular, they show that designing a channel with minimum
leakage is equivalent to Nash equilibria in a corresponding 
two-player zero-sum games of incomplete information for a range of 
entropy measures.

Concerning costs of security,
Yang et al.~\cite{Yang:12:POST} propose a 
framework to analyze user behavior in anonymity networks. 
Utility is modeled as a combination of weighted cost and anonymity utility. 
They also consider incentives and their impact on users' cooperation.

%There has also been experimental work in application of game theory for security.
%Grossklags and Reitter~\cite{Grossklags:14:CSF} investigate how 
%individuals’ cognitive predispositions impact on timing-related security choices.
%They demonstrate this by experiments based on a security game of timing.

Some security games have considered leakage of information about the 
defender's choices.
%, but in general they model are different from ours in that we model 
% a system that generates a single secret independently of the defender's strategy.
%
For example, Alon et al.~\cite{Alon:13:SIAMDM} present two-player zero-sum games where a defender chooses probabilities of secrets while an attacker chooses and learns some of the defender's secrets.
Then they show how the leakage on the defender's secrets influences the defender's optimal strategy.
%
%[Alon+, ICS 2010, SIDMA 2013]  "Adversarial Leakage in Games"~\cite{Alon:13:SIAMDM} \\
%Defender MAX chooses a probability distribution on secret and Attacker MIN  chooses a function (which we can regard as a deterministic channel that maps each secret to full/non-leakage).
%\url{http://www.tau.ac.il/~nogaa/PDFS/leakage7.pdf}
%
Xu et al.~\cite{Xu:15:IJCAI} present zero-sum security games where the attacker 
acquires partial knowledge on the security resources the defender is protecting, and show 
the defender's optimal strategy under such attacker's knowledge.
%
%[Xu+, IJCAI 2015]  "Security Games with Information Leakage: Modeling and Computation"~\cite{Xu:15:IJCAI} \\
%Defender is able to choose k out of N targets and protect them.
%Which target is protected (i.e., partial information on Defender's pure strategies) is leaked to the attacker with some probabilities.
%Attacker is able to choose a subset of the targets.
%Attacker's utility = sum of reward (when succeeded) and cost (when failed)
%Zero-sum game
%Compute the defender's optimal strategy under such information leakage on defender's pure strategy
%Their formalization do not use channels in information theory, and only compute probabilities
%\url{https://arxiv.org/pdf/1504.06058.pdf}
%
More recently, Farhang et al.~\cite{Farhang:16:GameSec} present two-player games 
where utilities are defined taking account of information leakage, although the defender's 
goal is different from our setting.
They consider a model where the attacker incrementally and stealthily obtains 
partial information on a secret, while the defender periodically changes the 
secret after some time to prevent a complete compromise of the system.
In particular,  the defender is not attempting to minimize the leak of a certain secret, but only to make it useless (for the attacker).  
Hence their model of defender and utility  is totally different from ours.
%
%[Farhang+, GameSec 2016]  "FlipLeakage: A Game-Theoretic Approach to Protect Against Stealthy Attackers in the Presence of Information Leakage"~\cite{DBLP:conf/gamesec/2016} \\
%Attacker obtains partial information leakage on a secret little by little when time passes.
%Defender changes the secret after some time passes to prevent full leakage of the secret.
%\url{https://link.springer.com/chapter/10.1007/978-3-319-47413-7_12}
%
%
To the best of our knowledge there have been no works exploring games with utilities defined as information-leakage measures.

Finally, in game theory Matsui ~\cite{Matsui:89:GEB} uses the term ``information leakage game'' 
with a meaning different than ours,
%to denote a different notion of games that do not deal with leakage of secrets.
%In his terminology it denotes 
namely, as a game in which (part of) the strategy of one player may be leaked in advance to the other player, and the latter may revise his strategy based on this
knowledge. 

\section{Conclusion and future work}
\label{sec:conclusion}
In this paper we introduced the notion of information leakage games,
in which a defender and an attacker have opposing goals in optimizing
the amount of information leakage in a system.
In contrast to standard game theory models, in our games the utility of 
a mixed strategy is a convex function of the 
distribution of the defender's actions, rather than the expected value of  the  utilities of the pure strategies in the support.  
Nevertheless, the important properties of game theory, 
notably the existence of a Nash equilibrium, still hold for our zero-sum leakage games, 
and we provided algorithms to compute the corresponding optimal strategies 
for the attacker and the defender. 

As future research, we would like to extend leakage games to scenarios
with repeated observations, i.e., when the attacker can repeatedly observe 
the outcomes of the system in successive runs, under the assumption that 
both the attacker and the defender may change the channel at each run. 
Furthermore, we would like to consider the possibility 
to adapt the defender's strategy to the secret value, as
we believe that in some cases this would provide 
significant advantage to the defender. 
We would also like to consider the cost of attack 
and of defense, which would lead to non-zero-sum games.

\paragraph*{Acknowledgments}
\begin{small}
The authors are thankful to Arman Khouzani and Pedro O. S. Vaz de Melo 
for valuable discussions.
This work was supported by JSPS and Inria under the project LOGIS of the Japan-France AYAME Program,
and by the project Epistemic Interactive Concurrency (EPIC) from the STIC 
AmSud Program.
M\'{a}rio S. Alvim was supported by CNPq, CAPES, and FAPEMIG.
Yusuke Kawamoto was supported by JSPS KAKENHI Grant Number JP17K12667.
\end{small}

\bibliographystyle{abbrv}
\bibliography{short,new}

\appendix

%\newpage
%\section{Reviews from GameSec}
%\input{reviews}

%\input{password-example}

%\input{working-notes}

\end{document}